\documentclass[10pt,a4paper]{article}

\usepackage{amssymb}
\usepackage{gastex}
\usepackage{enumerate}

\newcounter{mycount}[section]

\newtheorem{theorem}[mycount]{Theorem}
\newtheorem{lemma}[mycount]{Lemma}
\newtheorem{proposition}[mycount]{Proposition}
\newtheorem{corollary}[mycount]{Corollary}

\newtheorem{claim}{Claim}

{\vspace{0.7 em}}

\newenvironment{proof}
{\vspace{0.3 em} \noindent {\bf Proof.}}{\QED {\vspace{1.2 em}}}

\def\squareforqed{$\blacksquare$}
\def\QED{\ifmmode\squareforqed\else{\unskip\nobreak\hfil
\penalty50\hskip1em\null\nobreak\hfil\squareforqed
\parfillskip = 0pt\finalhyphendemerits = 0\endgraf}\fi}

\def\AND{\wedge}
\def\OR{\vee}

\def\bbB{\mathbb{B}}

\def\bbN{\mathbb{N}}

\def\ra{\rightarrow}

\def\spec{{\rm{spec}}}

\def\stra{\mathfrak{A}}
\def\strb{\mathfrak{B}}
\def\floor #1{\lfloor #1 \rfloor}

\def\Rels{\mathcal{R}}
\def\rev #1{\overleftarrow{#1}}
\def\srev #1{\scriptsize \overleftarrow{#1}}

\def\cC{\mathcal{C}}
\def\cF{\mathcal{F}}
\def\cM{\mathcal{M}}
\def\cP{\mathcal{P}} 

\def\cS{\mathcal{S}}
\def\cT{\mathcal{T}}

\def\vc{\bar{c}}
\def\vd{\bar{d}}

\def\vt{\bar{t}}

\def\vK{\bar{K}}
\def\vL{\bar{L}}
\def\vM{\bar{M}}
\def\vN{\bar{N}}
\def\vX{\bar{X}}
\def\vY{\bar{Y}}

\def\vone{\bar{1}}

\def\cR{\mathcal{R}}

\def\cone{\textsc{co-NE}}
\def\ne{\textsc{NE}}
\def\conp{\textsc{co-NP}}
\def\np{\textsc{NP}}
\def\p{\textsc{P}}

\def\spectrum{\textsc{Spec}}

\def\id{\textsf{\bf I}}

\def\FOtwo{\mathsf{FO}^2}

\def\Ctwo{\mathsf{C}^2}

\def\ESOC{\exists\mathsf{SO}\mathsf{C}^2}
\def\MLC{\mathsf{MLC}}
\def\QMLC{\mathsf{QMLC}}

\def\parikh{{\rm\textsf{Image}}}
 
\def\tp{\textsf{type}}

\def\bireg{\textsf{BiREG}}

\def\gbireg{\widetilde{\textsf{BiREG}}}

\def\biregc{\textsf{COMP-BiREG}}

\def\direg{\textsf{REG}}
\def\diregc{\textsf{COMP-REG}}

\def\const{\textsf{CON}}
\def\preb{\textsf{PREB}}
\def\prebatom{\textsf{PREB-Atom}}

\def\FOthree{\mathsf{FO}^3}
\def\CtwoT{\mathsf{C^2(<)}}

\def\indeg{\textrm{in-deg}}
\def\outdeg{\textrm{out-deg}}

\def\col{\textsf{col}}

\def\ugeq{^{\blacktriangleright}\!}
\def\ugeqN{{^{\blacktriangleright}\!\bbN}}

\def\OMIT #1{}

\title{Regular graphs and the spectra of two-variable logic with counting}

\author{Eryk Kopczy\'nski\thanks{University of Warsaw, \texttt{erykk@mimuw.edu.pl}}
\and
Tony Tan\thanks{Hasselt University and Transnational University of Limburg,
\texttt{ptony.tan@gmail.com}}}

\begin{document}

\date{}

\maketitle

\begin{abstract}
The {\em spectrum} of a first-order logic sentence 
is the set of natural numbers that are cardinalities of its finite models.
In this paper we show that when restricted to using only two variables,
but allowing counting quantifiers,
the spectra of first-order logic sentences are semilinear
and hence, closed under complement.
At the heart of our proof are semilinear characterisations 
for the existence of regular and biregular graphs, 
the class of graphs in which there are a priori bounds on the degrees of the vertices.
Our proof also provides a simple characterisation of models 
of two-variable logic with counting --
that is, up to renaming and extending the relation names,
they are simply a collection of regular and biregular graphs.
\end{abstract}


\section{Introduction}
\label{s: intro}

The spectrum of a first-order sentence $\phi$, denoted by $\spectrum(\phi)$,
is the set of natural numbers that are cardinalities of finite models of $\phi$.
Or, more formally, $\spectrum(\phi) = 
\{n \mid \mbox{there is a model of} \ \phi \ \mbox{of size} \ n\}$.
A set is a {\em spectrum},
if it is the spectrum of a first-order sentence.

In this paper we consider the logic $\Ctwo$,
the class of first-order sentences using only two variables and allowing
counting quantifiers $\exists^k z \ \phi(z)$, where $k\geq 1$.
Semantically $\exists^k z \ \phi(z)$ means
there exist at least $k$ number of $z$'s such that $\phi(z)$ holds.
We prove that the spectra of $\Ctwo$ are precisely semilinear sets.
In fact, our proof also shows that
the family of models of a $\Ctwo$ formula
can be viewed as a collection of regular graphs.

\subsection*{Related works}
The notion of spectrum was introduced by Scholz in~\cite{Scholz52}
where he also asked whether there exists a 
necessary and sufficient condition for a set to be a spectrum.
Since its publication, Scholz's question and many of its variants
have been investigated by many researchers for the past 60 years.
One of the arguably main open problems in this area is 
the one asked in~\cite{Asser55},
known as {\em Asser's conjecture},
whether the complement of a spectrum is also a spectrum.

The notion of spectrum has a deep connection with complexity theory
as shown by Jones and Selman~\cite{JS74},
as well as Fagin~\cite{Fagin73} independently
that a set of integers is a spectrum if and only if
its binary representation is in $\ne$.
Hence, Asser's conjecture is equivalent to asking whether $\ne=\cone$.
It also immediately implies that if
Asser's conjecture is false,
i.e., there is a spectrum whose complement is not a spectrum,
then $\np\neq\conp$, hence $\p\neq \np$.
The converse implication is still open.
An interesting result in~\cite{Woods81} states that
if spectra are precisely rudimentary sets,
then $\ne=\cone$ and $\np\neq\conp$.\footnote{It should be noted that
the class of rudimentary sets corresponds precisely to {\em linear time hierarchy} --
the linear time analog of polynomial time hierarchy~\cite{Wrathall78}.}
There are a number of interesting connections between spectrum
and various models of computation such as RAM as well as intrinsic computational behavior.
See, for example,~\cite{Gr84,Gr85,Gr90,Lynch82,RS55}.
We refer the reader to~\cite{DJMM12} for a more comprehensive treatment on the spectra problem
and its history.

The logic $\Ctwo$ is not the first logic known to have semilinear spectra.
A well known Parikh theorem states the spectra of context-free languages are semilinear,
and closed under complementation.
Using the celebrated composition method, Gurevich and Shelah in~\cite{GS03} showed that
the spectra of monadic second order logic with {\em one unary function} are semilinear.
Compared to the work in~\cite{GS03},
note that $\Ctwo$ one can express 
a few unary functions, hence our result does not follow from~\cite{GS03}, and
neither theirs from ours since we are restricted to using only two variables.

In~\cite{FM04} Fischer and Makowsky show that the {\em many-sorted spectra} 
of the monadic second-order logic 
with {\em modulo counting} over structures with {\em bounded tree-width}
are semilinear.
Intuitively, the many-sorted spectra of a formula
are spectra which counts the cardinality of the unary predicates in 
the models of the formula,
instead of just counting the sizes of the models.
The semilinearity is obtained by reduction to regular tree languages and ``pumping'' argument. 
This result is orthogonal to ours, 
since structures expressible in $\Ctwo$ do not have bounded tree-width.
An example is $d$-regular graphs for $d\geq 3$.
Moreover, due to unbounded tree-width, 
it is very unlikely that one can apply some sort of 
``pumping'' or automata theoretic argument 
as in~\cite{FM04} to obtain the semilinearity of $\Ctwo$ spectra.

As far as we know, $\Ctwo$ is the first logic
known to have its spectra closed under complement
without any restriction on the vocabulary nor in the interpretation.
The result closest to ours is the one by \'E.~Grandjean in~\cite{Gr90} 
where he considers the spectra of first-order sentences using only one variable.
A similar result due to M.~Grohe and stated in~\cite{DJMM12},
says that for every Turing machine $M$,
there exists a first-order sentence $\phi_M$ using only {\em three variables}
such that $\spectrum(\phi_M) = \{t^2 \mid t \ \mbox{is the length of an accepting run of} \ M\}$.

To end our study of related work, we should mention that
the two-variable logic and many of its variants have been
extensively studied, with the focus being mainly on the satisfiability problem.
For more development in this direction,
we refer the readers to~\cite{Mortimer75,Vardi96,Otto97,GOR97,GO99,Gradel01,Pratt-Hartmann05,PST00,KMPT12,ST13}
and the references therein.

\subsection*{Sketch of our proof}
Consider the following instances of structures expressible in $\Ctwo$.\footnote{Though 
the result in this paper holds for arbitrary structures, it helps
to assume that the structures of $\Ctwo$ are graphs in which the vertices and
the edges are labelled with a finite number of colours.}
\begin{enumerate}[(Ex.1)]\itemsep=0pt
\item
$(c,d)$-biregular graphs: the bipartite graphs on the vertices $U\cup V$,
where the degree of each vertex in $U$ and $V$ is $c$ and $d$, respectively.
\item
$(c,d)$-regular digraphs:
the directed graphs in which 
the in-degree and the out-degree of each vertex is $c$ and $d$, respectively. 
\end{enumerate}
An observation from basic graph theory tells us that 
for ``big enough'' $M$ and $N$,\footnote{``Big enough''
means $M$ and $N$ are greater than a constant $K$ which depends only on $c$ and $d$.}
\begin{enumerate}[(C1)]\itemsep=0pt
\item
there is a $(c,d)$-biregular graph in which $M$ vertices are of degree $c$ and
$N$ vertices of degree $d$ if and only if $Mc = Nd$;
\item
there is a $(c,d)$-regular digraph of $N$ vertices
if and only if $Nc = Nd$, and hence, $c=d$.
\end{enumerate}
These characterisations immediately imply that
the spectra of the formulas (Ex.1) and (Ex.2) above
are linear sets.
It is from these observations that we draw our inspiration
to prove the semilinearity of the spectra of $\Ctwo$.

More precisely, we show that given a $\Ctwo$ sentence $\varphi$,
one can construct a {\em Presburger} formula $\psi$
that expresses precisely the spectrum of $\varphi$.
Presburger formulas are first-order formulas 
with the relation symbols $+$ and $\leq$
and constants $0$ and $1$ interpreted over the domain $\bbN$ in the natural way.
It is shown by Ginsburg and Spanier in~\cite{GS66}
that Presburger formulas express precisely the class of semilinear sets.
That is, if $\psi(\vX)$ is a Presburger formula with free variables $\vX = (X_1,\ldots,X_k)$,
the set $\{\vN \in \bbN^{k} \mid \varphi(\vN) \ \mbox{holds}\}$ is semilinear.

The crux of our construction of the Presburger formula is a generalisation 
of the characterisations (C1)--(C2) above to the following setting.
Let $\cC$ be a set of $\ell$-colors, denoted by $\col_1,\col_2,\ldots,\col_{\ell}$,
and let $C$ and $D$ be $(\ell\times m)$- and $(\ell\times n)$-matrices 
whose entries are all non-negative integers.
We say that a bipartite graph 
$G=(U,V,E)$ is $(C,D)$-biregular, 
if we can color its edges with colors from $\cC$
such that there is a partition $U = U_1\cup \cdots \cup U_m$
and $V = V_1\cup\cdots\cup V_n$ where 
\begin{itemize}\itemsep=0pt
\item
for every integer $1\leq i \leq m$, 
for every vertex $u \in U_i$,
for every $1\leq j \leq \ell$,
the number of edges with color $\col_j$ adjacent to $u$ is {\em precisely} $C_{j,i}$; and
\item
for every integer $1\leq i \leq n$, 
for every vertex $v \in V_i$,
for every $1\leq j \leq \ell$,
the number of edges with color $\col_j$ adjacent to $v$ is {\em precisely} $D_{j,i}$.
\end{itemize}
Our setting also allows us to say that
the number of edges with color $\col_j$ adjacent to $v$ is {\em at least} $D_{j,i}$.
In Theorem~\ref{t: l-type biregular} we effectively construct a Presburger formula
that characterises the set 
$\{N \mid \mbox{there is a} \ (C,D)\mbox{-biregular graph of} \ N \ \mbox{vertices}\}$.

In a similar manner, we can define $(C,D)$-regular digraphs,
where $C$ and $D$ control the number of incoming and outgoing edges of each vertex, respectively.
Likewise, we obtain a similar Presburger formula 
that characterises the set 
$\{N \mid \mbox{there is a} \ (C,D)\mbox{-regular digraph of} \ N \ \mbox{vertices}\}$.\footnote{Closely 
related to our result
is the work by S.~L.~Hakimi~\cite{Hakimi62} which deals with the question: 
given a vector $(d_1,\ldots,d_m)$, 
is there a graph with vertices $v_1,\ldots,v_m$ 
whose degrees are precisely $d_1,\ldots,d_m$, respectively?
Another related result concerns the notion of score sequence 
obtained by H.~G.~Landau~\cite{Landau53}
which deals with the question:
given a vector $(d_1,\ldots,d_m)$, 
is there a tournament with vertices $v_1,\ldots,v_m$ 
whose outdegrees are precisely $d_1,\ldots,d_m$, respectively?
These questions are evidently different from our characterisations provided in
Section~\ref{s: regular graphs}.}
We then proceed to observe that the relations in every model of a $\Ctwo$ formula
can be partitioned in such a way that 
every part forms a $(C,D)$-regular digraph, 
and every two parts a $(C,D)$-biregular graph.
In a sense this shows
that the models of $\Ctwo$ is simply a collection of regular graphs.
Applying the Presburger formula that characterises the existence of these regular graphs,
we obtain the semilinearity of the spectra of $\Ctwo$ formulae.

For the converse direction, it is not that difficult to show that
every semilinear set is also a spectrum of a $\Ctwo$ sentence.
Since semilinear sets are closed under complement,
this establishes the fact that
the spectra of $\Ctwo$ are closed under complement.
It can also be deduced immediately from our proof that
the many-sorted spectra of $\Ctwo$ are also semilinear.
Moreover, our result extends trivially to the class $\ESOC$,
the class of sentences of the form: 
$\exists R_1\cdots \exists R_m \ \phi$,
where $R_1,\ldots,R_m$ are second-order variables and $\phi$
is a $\Ctwo$ formula.
We simply regard $R_1,\ldots,R_m$ as part of the signature.

\subsection*{Outline of the paper}
This paper is organised as follows.
In Section~\ref{s: c2} we review the logic $\Ctwo$ and state our main result:
Theorem~\ref{t: presburger for c2} which states that 
every $\Ctwo$ spectrum can be express in a Presburger formula.
The proof of Theorem~\ref{t: presburger for c2} is rather complex.
So we present its outline in Section~\ref{s: plan}, before its details in Sections~4--8.
Finally we conclude with a few observations and future directions in Section~\ref{s: conclusion}.

\section{The logic $\Ctwo$}
\label{s: c2}

In this section we review the definition of $\Ctwo$
and mention the main result in this paper
and its corollaries.
We fix 
$\cP = \{P_1,P_2,\ldots\}$ to be the set of predicate symbols of arity 1;
and
$\cR = \{R_1,R_2,\ldots\}$ the set of predicate symbols of arity 2.
Two-variable logic with counting, denoted by $\Ctwo$,
is defined by the following syntax.
\begin{eqnarray*}
\phi & ::= &
\begin{array}{c|c|c|c|c|c}
z=z & R(z,z) & P(z) & \neg \phi & \phi \AND \phi & \exists^{k} z\ \phi,
\end{array}
\end{eqnarray*}
where the variable $z$ ranges over $x,y$,
and the symbols $R$ and $P$ over $\cR$ and $\cP$, respectively.

The quantifier $\exists^k z\ \phi$ means
\emph{there are at least $k$ elements $z$} such that $\phi$ holds.
Note that $\exists^1 z\ \phi$ is the standard $\exists z\ \phi$, 
and $\forall z \ \phi$ is equivalent to $\neg \exists^1 z \ \neg \phi $.
By default, we assume that $\exists^0 z\ \phi$ always holds.

As usual, we write $\stra \models \phi$ to denote that
the structure $\stra$ is a model of $\phi$
and $\spectrum(\phi)$ to denote the spectrum of $\phi$.
Theorem~\ref{t: presburger for c2} below
is the main result in this paper.
Its proof spans over Sections~4--8.
\begin{theorem}
\label{t: presburger for c2}
For every $\phi \in \Ctwo$,
there exists a Presburger formula $\preb(x)$
such that the set $\{n \mid \preb(n) \ \mbox{holds}\} = \spectrum(\phi)$.
Moreover, the formula $\preb(x)$ can be constructed effectively.
\end{theorem}

We should remark that Theorem~\ref{t: presburger for c2} also holds
for arbitrary vocabulary.
Since $\Ctwo$ uses only two variables,
relations of greater arity such as $R(x,y,x,x,y)$ 
can be viewed simply as unary or binary relations;
so we can create new binary and unary relations for
each possible combination, and 
easily verify whether the result is consistent.

An immediate consequence of Theorem~\ref{t: presburger for c2}
is the spectra of $\Ctwo$ are semilinear.
\begin{corollary}
\label{c: c2 semilinear}
For every sentence $\phi \in \Ctwo$,
the spectrum $\spectrum(\phi)$ is semilinear.
\end{corollary}

On the other hand, it is not that difficult to show that
every semilinear set is a spectrum of a $\Ctwo$ sentence,
as formally stated below.
\begin{proposition}
\label{p: semilinear to c2}
For every semilinear set $\Lambda \subseteq \bbN$,
there exists a sentence $\phi \in \Ctwo$ such that $\spectrum(\phi)=\Lambda$.
\end{proposition}
\begin{proof}
For a linear set $\Gamma_{m,n} = \{ m+ in \mid i =0,1,2,\ldots\}$,
consider the vocabulary $\tau_{m,n} = \{A,B_0,B_1,\ldots,B_{n-1},E\}$,
where $A,B_0,\ldots,B_{n-1}$ are unary and $E$ binary.
Consider the $\Ctwo$ sentence $\phi_{m,n}$ which states that
$A\cup B_0 \cup \cdots \cup B_{n-1}$ partition the whole universe,
the predicate $A$ contains exactly $m$ elements,
and for every $x$, if $B_i(x)$ holds, 
\begin{itemize}\itemsep=0pt
\item
there is exactly one $y$ such that $x\neq y$ and $B_{i+1 \ {\rm mod} \ n}(y)$ and $E(x,y)$ holds;
\item
there is exactly one $y$ such that $x\neq y$ and $B_{i-1\ {\rm mod} \ n}(y)$ and $E(y,x)$ holds.
\end{itemize}
It is straightforward that $\spectrum(\phi_{m,n}) = \Gamma_{m,n}$.
For a semilinear set, we simply takes the finite union of such $\phi_{m,n}$.
This completes our proof of Proposition~\ref{p: semilinear to c2}.
\end{proof}

Now, take Corollary~\ref{c: c2 semilinear},
apply the fact that semilinear sets are closed under complement,
and then Proposition~\ref{p: semilinear to c2},
we obtain the following corollary.

\begin{corollary}
\label{c: spectra c2 closed complement}
The spectra of $\Ctwo$ sentences are closed under complement within $\Ctwo$.
\end{corollary}

\section{The plan for the proof of Theorem~\ref{t: presburger for c2}}
\label{s: plan}

As mentioned earlier, the proof of Theorem~\ref{t: presburger for c2}
is rather complex and spans over Sections~4--8.
We give its outline here.
\begin{itemize}\itemsep=0pt
\item
In Section~\ref{s: qmlc} we define the logic $\QMLC$ (Quantified Modal Logic with Counting),
which for our purpose, will be easier to work with.
In particular, we show that $\Ctwo$ and $\QMLC$ are equivalent in terms of spectra.
\item
In Section~\ref{s: regular graphs} we define 
the class of biregular graphs and regular digraphs.
The main theorems in this section are 
Theorems~\ref{t: main biregular} and~\ref{t: main regular-digraph}.
Theorem~\ref{t: main biregular} gives us the Presburger characterisations of the existence
of biregular graphs,
while Theorem~\ref{t: main regular-digraph} the same characterisations for the regular digraphs.
\item
Equipped with Theorems~\ref{t: main biregular} and~\ref{t: main regular-digraph},
we construct the formula $\preb(x)$ as required in Theorem~\ref{t: presburger for c2} 
in Section~\ref{s: proof}. 
\item
However, the proofs of Theorems~\ref{t: main biregular} and~\ref{t: main regular-digraph}
are themselves rather long and involved.
So we postpone the proof of Theorem~\ref{t: main biregular} until Section~\ref{s: proof1}.
The proof of Theorem~\ref{t: main regular-digraph} is similar to Theorem~\ref{t: main biregular},
so we simply sketch it in Section~\ref{s: proof2}.
\end{itemize}
Figure~\ref{fig:interdependence} illustrates the interdependence among Sections~4--8.

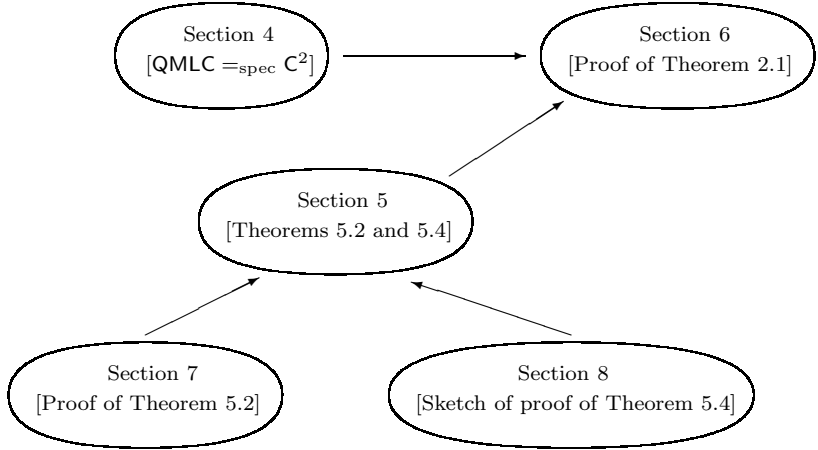
\begin{figure}[t]
{\footnotesize
\begin{picture}(130,70)(-65,-35)

\qbezier(18,22)(18,29)(36,29)
\qbezier(36,29)(54,29)(54,22)
\qbezier(54,22)(54,15)(36,15)
\qbezier(36,15)(18,15)(18,22)
\put(31,24){Section~6}
\put(21.5,20){[Proof of Theorem~\ref{t: presburger for c2}]}

\qbezier(-38,22)(-38,29)(-24,29)
\qbezier(-24,29)(-10,29)(-10,22)
\qbezier(-10,22)(-10,15)(-24,15)
\qbezier(-24,15)(-38,15)(-38,22)
\put(-8,22){\vector(1,0){24}}
\put(-29,24){Section~4}
\put(-34,20){[$\QMLC =_{\rm spec} \Ctwo$]}

\qbezier(-27,0)(-27,7)(-9,7)
\qbezier(-9,7)(9,7)(9,0)
\qbezier(9,0)(9,-7)(-9,-7)
\qbezier(-9,-7)(-27,-7)(-27,0)
\put(6,6){\vector(3,2){15}}
\put(-14,2){Section~5}
\put(-23.5,-2){[Theorems~\ref{t: main biregular} and~\ref{t: main regular-digraph}]}

\qbezier(-52,-23)(-52,-16)(-34,-16)
\qbezier(-34,-16)(-16,-16)(-16,-23)
\qbezier(-16,-23)(-16,-30)(-34,-30)
\qbezier(-34,-30)(-52,-30)(-52,-23)
\put(-34,-15){\vector(2,1){15}}
\put(-39,-21){Section~7}
\put(-48.5,-25){[Proof of Theorem~\ref{t: main biregular}]}

\qbezier(-2,-23)(-2,-16)(22,-16)
\qbezier(22,-16)(46,-16)(46,-23)
\qbezier(46,-23)(46,-30)(22,-30)
\qbezier(22,-30)(-2,-30)(-2,-23)
\put(22,-15){\vector(-3,1){21}}
\put(15,-21){Section~8}
\put(1.5,-25){[Sketch of proof of Theorem~\ref{t: main regular-digraph}]}

\end{picture}}
\caption{The skeleton for the proof of Theorem~\ref{t: presburger for c2}.
Theorem~\ref{t: main biregular} gives us the Presburger characterisations for
the existence of biregular graphs, while Theorem~\ref{t: main regular-digraph}
for the regular digraphs.}
\label{fig:interdependence}
\end{figure}

\section{Quantified modal logic with counting}
\label{s: qmlc}

In this section we present  {\em quantified modal logic with counting} ($\QMLC$),
which for our purpose, will be easier to work with.
We are going to show that 
$\Ctwo$ and $\QMLC$ are equivalent in terms of spectra.
In fact, our proof shows that 
$\Ctwo$ and $\QMLC$ are equivalent up to renaming/deleting/adding relational symbols,
when $\QMLC$ are restricted to ``complete'' structures defined as follows.
A structure $\stra$ is a complete structure, if it satisfies the following properties.
\begin{enumerate}[(N1)]\itemsep=0pt
\item
$\stra$ is a clique over $A$.
That is, for every $a,b \in A$, either $a=b$ or
$R(a,b)$ for some $R \in \Rels$.
\item
Every binary relation in $\Rels$ does not intersect identity relation.
That is, for every $R \in \Rels$, if $R(a,b)$ holds, $a\neq b$.
\item
$\Rels$ is closed under inverse.
That is, for every $R \in \Rels$, there exists $\rev{R}\in \Rels$
such that $\rev{R}\neq R$ and for every $a,b\in A$,
$R(a,b)$ if and only if $\rev{R}(b,a)$.
\item
The binary predicates in $\Rels$ are pairwise disjoint.
\end{enumerate}
Our proof is an adaptation of the proof in~\cite{LSW01} 
which shows that similar equivalence holds 
between two-variable logic and modal logic.


The class $\MLC$ of {\em modal logic with counting}
is defined with the following syntax.
\begin{eqnarray*}
\phi & ::= &
\begin{array}{c|c|c|c}
\neg \phi & \alpha & \phi \AND \phi & \Diamond_{\beta}^k \phi
\end{array}
\end{eqnarray*}
where $\alpha$ ranges over $\cP$ and $\beta$ over $\cR$.

The semantics of $\MLC$ is as follows.
Let $\stra$ be a structure of $\tau$
and $a\in A$ and $\phi$ be an $\MLC$ formula.
That $\stra$ satisfies $\phi$ from $a$,
denoted by $\stra,a \models \phi$, is defined as follows.
\begin{itemize}\itemsep=0pt
\item
$\stra,a \models P$, where $P \in \cP$,
if $P(a)$ holds in $\stra$.
\item
$\stra,a \models \neg \phi$, 
if $\stra,a \not\models \phi$.
\item
$\stra,a \models \phi_1 \AND \phi_2$, 
if $\stra,a \models \phi_1$ and $\stra,a \models \phi_2$.
\item
$\stra,a \models \Diamond_{R}^k \phi$,
if there exist at least $k$ elements $b_1,\ldots,b_k \in A$
such that $R(a,b_i)$ holds in $\stra$ and $\stra,b_i \models \phi$
for $i=1,\ldots,k$.
\end{itemize}

We define the class of quantified modal logic with counting, 
denoted by $\QMLC$ with the following syntax.
\begin{eqnarray*}
\psi & ::= &
\begin{array}{c|c|c}
\neg \psi & \psi_1 \AND \psi_2 & \exists^k \phi
\end{array}
\end{eqnarray*}
where the formula $\phi \in \MLC$.
A $\QMLC$ formula $\psi$ is called a {\em basic} $\QMLC$,
if it is of the form $\exists^k \ \phi$,
where $\phi \in \MLC$.

The semantics of $\QMLC$ is as follows.
Let $\stra$ be a structure of $\tau$ and $\psi\in\QMLC$.
That $\stra$ satisfies $\psi$, denoted by $\stra \models \psi$, is defined as follows.
\begin{itemize}\itemsep=0pt
\item
$\stra \models \neg \psi$, 
if it is not the case that $\stra \models \psi$.
\item
$\stra \models \psi_1 \AND \psi_2$, 
if $\stra \models \psi_1$ and $\stra \models \psi_2$.
\item
$\stra \models \exists^k \phi$,
if there exist at least $k$ elements $a_1,\ldots,a_k \in A$
such that $\stra,a_i \models \phi$ for $i=1,\ldots,k$.
\end{itemize}
We denote by $\spectrum(\psi)$
the set consists of the size of complete structures of $\psi$.
That is, 
$\spectrum(\psi) = \{n \mid \mbox{there is a complete structure} \ \stra\models \psi \ \mbox{of size} \ n\}$.
Note that for $\QMLC$, the notion of spectrum is restricted to complete structures.

In the following we are going to show that from spectral point of view,
$\Ctwo$ and $\QMLC$ are equivalent.
The intuitive explanation for the requirement of complete structure is as follows.
Notice that in $\QMLC$ we cannot express the negation of a binary relation $\neg R(x,y)$.
Rather, to ``express'' $\neg R(x,y)$ in $\QMLC$,
we introduce a new relation symbol to capture $\neg R(x,y)$,
hence, the requirement (N1) and (N4) in the complete structure.
Similarly, in $\QMLC$ from an element $x$,
we cannot express the ``inverse'' direction $R(y,x)$.
So we introduce a new relation $\rev{R}$ that captures the ``inverse'' of $R$,
and $R(y,x)$ will be simulated by $\rev{R}(x,y)$, instead,
hence the requirement (N3).
We require (N2) simply for technical convenience.

Theorem~\ref{t: c2 and qmlc} below states formally
the spectral equivalence between $\Ctwo$ and $\QMLC$,
when $\QMLC$ is restricted to complete structures.
\begin{theorem}
\label{t: c2 and qmlc}
For every $\varphi \in \Ctwo$,
there is a $\QMLC$ formula $\psi$ such that
\begin{itemize}\itemsep=0pt
\item
for every structure $\stra \models \varphi$,
there is a complete structure $\strb \models \psi$ where $|A|=|B|$;
\item
for every complete structure $\strb \models \psi$,
there is a structure $\stra \models \varphi$ and $|A|=|B|$.
\end{itemize}
\end{theorem}
\begin{proof}
Let $\varphi \in \Ctwo$.
By extending/renaming/deleting the relations, 
and by modifying the sentence $\varphi$, if necessary,
we can obtain another $\Ctwo$ sentence $\varphi'$ such that
\begin{itemize}\itemsep=0pt
\item
for every structure $\stra \models \varphi$,
there is a complete structure $\stra' \models \varphi'$ where $|A|=|A'|$;
\item
for every complete structure $\stra' \models \varphi'$, 
there is a structure $\stra \models \varphi$ and $|A|=|A'|$.
\end{itemize}
The details of the construction of $\varphi'$ is straightforward, hence, omitted.
For example, to achieve (N1) and (N4) we can introduce
a new binary relation for each Boolean combination of relations in $\varphi$.
We can do similar trick to achieve (N2) and (N3).

From this formula $\varphi'$, we are going to construct the desired $\QMLC$ formula $\phi$.
It consists of the following two steps.
\begin{enumerate}\itemsep=0pt
\item
Convert the sentence $\varphi'$ into its ``normal form'' $\psi$
such that for every complete structure $\stra$,
we have $\stra \models \varphi'$ if and only if $\stra \models \psi$.\footnote{We
would like to remark that the normal form here is different from
the standard Scott's normal form.}
\item
Convert the sentence $\psi$ into a ``quantified modal logic'' ($\QMLC$)
sentence $\phi$
such that for every complete structure $\stra$,
we have $\stra \models \psi$ if and only if $\stra \models \phi$.
\end{enumerate}
In the following paragraphs
we are going to describe formally these two steps.

A $\Ctwo$ sentence is in {\em normal form}, if
all the quantifiers are either of form 
\[
\exists^k y \ \Big( R(x,y) \AND \theta(y) \Big),
\qquad\mbox{or}\qquad
\exists^k x \ \theta(x)
\] 
and all other applications of variables are of form $P(x)$, where $P$ is a unary predicate.

The $\Ctwo$ sentence $\varphi'$ can be converted
into its equivalent sentence $\psi$ in normal form as follows.
\begin{itemize}\itemsep=0pt
\item
First, we rewrite every subformula of the form
$\exists^k y\ \theta(x,y)$ with one free variable $x$
into the following form:
\begin{eqnarray*}
&  \theta(x,x) \ \AND  \ \exists^{k-1} y\ \Big((x\neq y) \AND \theta(x,y)\Big) &
\\
& \OR &
\\
& \exists^k y\ \Big((x\neq y) \AND \theta(x,y)\Big) &
\end{eqnarray*}
After such rewriting, we can assume that every quantifier in $\varphi$ is of the form 
$\exists^{k} y\ ((x\neq y) \AND \theta(x,y))$.
\item
Second, every quantification $\exists^{k} y\ ((x\neq y) \AND \theta(x,y))$,
in which $\theta(x,y)$ contains a subformula $\alpha(x)$ depending only on $x$,
can be rewritten into the form:
\begin{eqnarray*}
& \neg\alpha(x) \ \AND \ \exists^k y\ \Big((x\neq y) \AND \theta_0(x,y)\Big) &
\\
& \OR & 
\\
& \alpha(x) \ \AND \ \exists^k y\ \Big((x\neq y) \AND \theta_1(x,y)\Big) &
\end{eqnarray*}
where $\theta_0(x,y)$ and $\theta_1(x,y)$ are obtained from $\theta$ 
by replacing $\alpha(x)$ with false and true, respectively.
We can repeat this until $\theta(x,y)$ no longer has a subformula depending only on $x$.
\\
After such rewriting we can assume that every quantifier in $\varphi$ is of the form 
\[
\exists^{k} y\ ((x\neq y) \AND \theta(x,y)),
\]
where $\theta(x,y)$ does not contain any subformula depending only on $x$.
\item
Third, every quantification $\exists^{k} y\ \Big((x\neq y) \AND \theta(x,y)\Big)$
can be rewritten into the form:
\[
\bigvee_{f \in \Delta_\Rels^{k}} \ \
\bigwedge_{R \in \Rels} \ \
\exists^{f(R)} y \  \Big( R(x,y) \AND \theta_R(y) \Big)
\]
where $\Delta_\Rels^k$ is the set of all functions $f: \Rels \ra \bbN$ 
such that $\sum_{R \in \Rels} f(R) = k$, and 
$\theta_R(y)$ is obtained from $\theta(x,y)$ by replacing each $R'(x,y)$ with true if $R=R'$,
and false otherwise.
\end{itemize}
By performing these three steps, 
we get the $\Ctwo$ sentence $\psi$ in the normal form.
Particularly, for every complete structure $\stra$,
we have $\stra \models \varphi'$ if and only if $\stra \models \psi$.

Now from this $\Ctwo$ sentence $\psi$ in normal form,
the construction of its $\QMLC$ sentence $\phi= F(\psi)$ 
can be done inductively as follows.
There are two cases.
\begin{enumerate}\itemsep=0pt
\item
$\vartheta$ has no free variable.
\begin{itemize}\itemsep=0pt
\item
If $\vartheta$ is $\neg \vartheta_1$,
then $F(\vartheta) = \neg F(\vartheta_1)$.
\item
If $\vartheta$ is $\vartheta_1\AND\vartheta_2$,
then $F(\vartheta) = F(\vartheta_1) \AND F(\vartheta_2)$.
\item
If $\vartheta$ is $\exists^k x \ \vartheta_1(x)$,
then $F(\vartheta) = \exists^k F(\vartheta_1(x))$.
\end{itemize}
\item
$\vartheta$ has one free variable $x$.
\begin{itemize}\itemsep=0pt
\item
If $\vartheta(x)$ is $P(x)$,
then $F(\vartheta(x)) = P$.
\item
If $\vartheta(x)$ is $\vartheta_1(x)\AND \vartheta_2(x)$,
then $F(\vartheta(x)) = F(\vartheta_1(x))\AND F(\vartheta_2(x))$. 
\item
If $\vartheta(x)$ is $\neg \vartheta_1(x)$,
then $F(\vartheta(x)) = \neg F(\vartheta_1(x))$. 
\item
If $\vartheta(x)$ is $\exists^k y \ R(x,y) \AND \vartheta_1(x,y)$,
then $F(\vartheta(x)) = \Diamond_{R}^k F(\vartheta_1(x,y))$. 
\end{itemize}
The case when $\vartheta$ has one free variable $y$ can be handled
in a symmetrical way.
\end{enumerate}
By a straightforward induction,
we can show that for every compete structure $\stra$,
$\stra \models \vartheta$ if and only if
$\stra\models F(\vartheta)$.
In particular, from the equivalences between $\varphi$ and $\varphi'$,
between $\varphi'$ and $\psi$ as well as
between $\psi$ and $\phi = F(\psi)$,
we obtain that
\begin{itemize}\itemsep=0pt
\item
for every structure $\stra \models \varphi$,
there is a complete structure $\strb \models \psi$ such that $|A|=|B|$;
\item
for every complete structure $\strb \models \psi$, 
there is a structure $\stra \models \varphi$ such that $|A|=|B|$.
\end{itemize}
This concludes our proof of Theorem~\ref{t: c2 and qmlc}.
\end{proof}


\section{Regular graphs}
\label{s: regular graphs}

In this section we are going to introduce two types of regular graphs:
biregular graphs (bipartite regular graphs)
and regular digraphs.
The main results in this section are
Theorems~\ref{t: main biregular} and~\ref{t: main regular-digraph},
which will be used in our proof of Theorem~\ref{t: presburger for c2}.
For the sake of readability, we postpone their proofs
until Sections~\ref{s: proof1} and~\ref{s: proof2}.

\subsection{Biregular graphs}
\label{ss: biregular}

An $\ell$-type bipartite graph is $G = (U,V,E_1,\ldots,E_{\ell})$,
where $E_1,\ldots,E_{\ell}$ are {\em pairwise disjoint} 
subsets of $U\times V$.
Elements in $E_i$ are called $E_i$-edges.
It helps to think of $G$ as a bipartite graph 
in which the edges are coloured with $\ell$ number of colours.

For a vertex $u \in U \cup V$, $\deg_{E_i}(u)$ denotes 
the number of $E_i$-edges adjacent to it, and
$\deg(u) = \sum_{i=1}^{\ell} \deg_{E_i}(u)$.
We write $\deg(G) = \max\{\deg(u) \mid u \ \mbox{is a vertex in} \ G \}$.
For an integer $d \in \bbN$,
we write $\deg_{E_i}(u) = {\ugeq d}$,
to denote $\deg_{E_i}(u) \geq d$.

Let $\bbN$ denote the set of natural numbers $\{0,1,2,\ldots\}$ 
and ${\ugeq \bbN} = \{\ugeq 0, \ugeq 1,\ugeq 2,\ldots\}$
and $\bbB = \bbN \cup {\ugeq\bbN}$.
We write $\bbB^{\ell\times m}$ to denote 
the set of $\ell\times m$ matrices whose entries are from $\bbB$.
The entry in row $i$ and column $j$ of a matrix $D \in \bbB^{\ell\times m}$
is denoted by $D_{i,j}$.

Let $C\in \bbB^{\ell \times m}$ and $D \in \bbB^{\ell \times n}$.
An $\ell$-type bipartite graph 
$G=(U,V,E_1,\ldots,E_{\ell})$ is $(C,D)$-biregular,
if there is a partition $U = U_1 \cup \cdots \cup U_m$
and $V = V_1 \cup \cdots \cup V_n$ such that
the following holds.
\begin{itemize}\itemsep=0pt
\item
For every $i = 1,\ldots,\ell$,
for every $j=1,\ldots,m$,
for every vertex $u \in U_j$, $\deg_{E_i}(u)= C_{i,j}$.
\item
For every $i = 1,\ldots,\ell$,
for every $j=1,\ldots,n$,
for every vertex $v \in V_j$, $\deg_{E_i}(v)= D_{i,j}$.
\end{itemize}
We call the partitions $U = U_1 \cup \cdots \cup U_m$
and $V = V_1 \cup \cdots \cup V_n$
the witness of the $(C,D)$-biregularity of $G$.
We say that the $(C,D)$-biregular graph $G$ 
is of size $(\vM,\vN)$, if 
$\vM = (|U_1|,\ldots,|U_m|)$ and $\vN = (|V_1|,\ldots,|V_n|)$.

\begin{theorem}
\label{t: l-type biregular}
For every two matrices $C \in \bbB^{\ell\times m}$
and $D \in \bbB^{\ell\times n}$,
there is a Presburger formula 
$\bireg_{C,D}(\vX,\vY)$,
where $\vX = (X_1,\ldots,X_m)$ and $\vY = (Y_1,\ldots,Y_n)$
such that the following holds.
There exists an $\ell$-type $(C,D)$-biregular graph of size $(\vM,\vN)$ 
if and only if
$\bireg_{C,D}(\vM,\vN)$ holds.
\end{theorem}

Theorem~\ref{t: l-type biregular} 
is then generalised to the case of complete bipartite graphs.
An $\ell$-type bipartite graph 
$G = (U,V, E_1,\ldots,E_{\ell})$ is {\em complete},
if $U\times V = E_1\cup \cdots\cup E_{\ell}$.
If $G$ is also a $(C,D)$-biregular graph,
then we call it a $(C,D)$-complete-biregular graph.

The following theorem is the main result in this subsection
that will be used in the proof in Section~\ref{s: proof}.
\begin{theorem}
\label{t: main biregular}
For every two matrices $C \in \bbB^{\ell\times m}$ and $D \in \bbB^{\ell\times n}$,
there is a Presburger formula $\biregc_{C,D}(\vX,\vY)$,
where $\vX = (X_1,\ldots,X_m)$ and $\vY = (Y_1,\ldots,Y_n)$ such that
the following holds.
There exists a $(C,D)$-complete-biregular graph of size $(\vM,\vN)$
if and only if $\biregc_{C,D}(\vM,\vN)$ holds. 
\end{theorem}


\subsection{Regular digraphs}
\label{ss: regular directed}


An $\ell$-type directed graph (or, digraph for short) 
is a tuple $G = (V,E_1,\ldots,E_{\ell})$,
where $E_1,\ldots,E_{\ell}$ are {\em pairwise disjoint} irreflexive relations on $V$
and for every $u,v \in V$,
if $(u,v) \in E_1 \cup \cdots \cup E_{\ell}$,
then the inverse direction $(v,u) \notin E_1 \cup \cdots \cup E_{\ell}$.
Edges in $E_i$ are called $E_i$-edges.

We will write $\indeg_{E_i}(u)$ to denote 
the number of incoming $E_i$-edges toward the vertex $u$, and
$\outdeg_{E_i}(u)$ 
the number of outgoing $E_i$-edges from the vertex $u$.
As before, for an integer $d \in \bbN$,
we write $\indeg_{E_i}(u) = {\ugeq d}$ and $\outdeg_{E_i}(u) = {\ugeq d}$,
to indicate that $\indeg_{E_i}(u) \geq d$ and
$\indeg_{E_i}(u) \geq d$, respectively.


Let $C,D \in \bbB^{\ell \times m}$.
An $\ell$-type digraph $G=(V,E_1,\ldots,E_{\ell})$ 
is $(C,D)$-regular-digraph,
if there exists a partition $V=V_1\cup\cdots\cup V_m$ 
such that for each $i=1,\ldots,\ell$,
for each $j = 1,\ldots,m$,
for each vertex $v \in V_j$,
$\indeg_{E_i}(v) = C_{i,j}$ and
$\outdeg_{E_i}(v) = D_{i,j}$.
We call $V_1\cup\cdots\cup V_m$ a witness of the $(C,D)$-regularity of $G$
and the graph $G$ is of size $\vN=(N_1,\ldots,N_m)$,
if $(N_1,\ldots,N_m)=(|V_1|,\ldots,|V_m|)$.

\begin{theorem}
\label{t: l-type direct-regular}
For every $C,D \in \bbB^{\ell\times m}$,
there exists a Presburger formula 
$\direg_{C,D}(\vX)$, where $\vX = (X_1,\ldots,X_m)$ such that the following holds.
There exists a $(C,D)$-regular-digraph of size $\vN$ 
if and only if
$\direg_{C,D}(\vN)$ holds.
\end{theorem}

Similar to Section~\ref{ss: biregular},
Theorem~\ref{t: l-type direct-regular} will be generalised
to the case of complete regular digraph.
An $\ell$-type graph $G = (V, E_1,\ldots,E_{\ell})$ is a {\em complete} digraph,
if for every two different vertices $u,v$,
either $(u,v)$ or $(v,u)$ is in $E_1\cup\cdots\cup E_{\ell}$.
If $G$ is also a $(C,D)$-regular,
then we call $G$ a $(C,D)$-complete-regular digraph.

The following theorem is the main result in this subsection
that will be used in the proof in Section~\ref{s: proof}.

\begin{theorem}
\label{t: main regular-digraph}
For every $C,D \in \bbB^{\ell\times m}$,
there exists a Presburger formula $\diregc_{C,D}(\vX)$, where $\vX = (X_1,\ldots,X_m)$
such that the following holds.
There exists a $(C,D)$-complete-regular digraph of size $\vN$
if and only if $\diregc_{C,D}(\vN)$ holds. 
\end{theorem}

\section{Proof of Theorem~\ref{t: presburger for c2}}
\label{s: proof}

Now we are ready to prove Theorem~\ref{t: presburger for c2}.
Let $\phi$ be a $\QMLC$ sentence.
Recall that a {\em basic} $\QMLC$ formula is of the form $\exists^k \ \varphi$,
where $\varphi \in \MLC$.
We also assume that in $\phi$ we have ``pushed'' all the negations inside
so that they are applied only to basic $\QMLC$.
We are going to construct a Presburger formula $\preb_{\phi}(x)$
such that $\spectrum(\phi) = \{n \mid \preb_{\phi}(n) \ \mbox{holds}\}$.

Before we proceed, we need a few auxiliary notations.
Let $\cP$ be the set of unary predicates used in $\phi$
and $\cR =\{R_1,\ldots,R_{\ell},\rev{R}_1,\ldots,\rev{R}_{\ell}\}$ 
the set of binary relations used in $\phi$,
where $\rev{R}_i$ is the inverse relation of $R_i$.
Let $K$ be the integer such that for all 
subformulae $\Diamond_R^l \psi$ in $\phi$,
we have $l \leq K$.

We denote by $\cM_{\phi}$ the set of all $\MLC$ subformulae of $\phi$
and their negations. 
A {\em type} in $\phi$ is a subset $T \subseteq \cM_{\phi}$ such that
\begin{itemize}\itemsep=0pt
\item
if $\varphi_1\wedge \varphi_2 \in T$,
then both $\varphi_1,\varphi_2 \in T$;
\item
$\varphi\in T$ if and only if $\neg \varphi \notin T$;
\item
if $\neg (\varphi_1\wedge \varphi_2) \in T$,
then at least one of $\neg \varphi_1,\neg\varphi_2 \in T$.
\end{itemize}
For a structure $\stra$ (not necessarily a model of $\phi$) and an element $a\in A$,
we define {\em the type of $a$ in $\stra$}, 
denoted by $\tp_{\stra}(a) \subseteq \cM_{\phi}$,
where
$\varphi \in \tp_{\stra}(a)$ if and only if
$\stra,a \models \varphi$.
For a type $T$, 
we write $T(\stra)$ to denote the set of elements in $A$ with type $T$.
Note that the sets $T(\stra)$'s are pairwise disjoint.
We let $\cT_{\phi}$ to be the set of all types in $\phi$.

We say that a function $f\ : \ \cT_{\phi}\times\cR\times\cT_{\phi}
\to
\{0,1,\ldots,K\}
\ \cup \
\{{\ugeq K}\}$ is {\em consistent}, 
if for every $T \in \cT_{\phi}$ the following holds.
\begin{itemize}\itemsep=0pt
\item
If $\Diamond_R^l \ \mu \in T$, then
$\sum_{T' \ \mbox{\scriptsize s.t.} \ T' \ni \mu} \ f(T,R,T')  \geq  l$.
\item
If $\neg (\Diamond_R^l \ \mu) \in T$, then
$\sum_{T' \ \mbox{\scriptsize s.t.} \ T' \ni \mu} \ f(T,R,T')   \leq  l-1$,
and $f(T,R,T') \in \bbN$,
for every $R \in \cR$ and for every type $T' \ni \mu$.
\end{itemize}
In the following we enumerate the set of all consistent functions
$\cF = \{f_1,\ldots,f_m\}$, 
the set of all types in $\cT_{\phi}=\{T_1,\ldots,T_n\}$,
and the set $\cT_{\phi}\times\cF = \{(T_1,f_1),\ldots,(T_n,f_m)\}$.

The desired Presburger formula $\preb_{\phi}(x)$ is defined as
the formula:
\begin{eqnarray*}
& \exists X_{(T_1,f_1)}  \cdots \exists X_{(T_n,f_m)} &
\Bigg(x = \sum_{1\leq i \leq n} \ \ \sum_{1 \leq j \leq m} \ X_{T_i,f_j}\Bigg)
\; \wedge \; \prebatom_{\phi}(\vX)
\; \wedge\; 
\const(\vX)
\end{eqnarray*}
where $\vX = (X_{(T_1,f_1)},\ldots,X_{(T_n,f_m)})$ is
the vector of all the variables $X_{(T,f)}$'s.

The intended meaning of the variable $X_{T_i,f_j}$
and the formulas $\prebatom_{\phi}(x)$ and $\const(\vX)$ is as follows. 
The variable $X_{T_i,f_j}$ is to represent the number of elements
of type $T_i$ and for each binary relation $R \in \cR$
and a type $S \in \cT_{\phi}$,
there is $f(T_i,R,S)$ number of outgoing $R$-edges towards the elements of type $S$.

Naturally, the total number of all elements in the universe will be
the sum of all $X_{T_i,f_j}$'s, hence, the sum:
\begin{eqnarray*}
x & = & \sum_{1\leq i \leq n} \ \ \sum_{1 \leq j \leq m} \ \ X_{T_i,f_j}.
\end{eqnarray*}

The formula $\prebatom_{\phi}(x)$ is to make sure that
the satisfiability of the $\QMLC$ sentence is preserved.
Formally, it is defined inductively as follows.
(Recall that all the negations have been ``pushed'' inside
so that they are applied only to basic $\QMLC$.)
\begin{itemize}\itemsep=0pt
\item
If $\phi : = \exists^k \varphi$, 
then
$\prebatom_{\phi} := 
\sum_{(T,f) \ \mbox{\scriptsize s.t.} \ \varphi \in T} \ X_{(T,f)} \geq k$.

\item
If $\phi : = \neg \exists^k \varphi$, 
then
$\prebatom_{\phi} :=
\sum_{(T,f) \ \mbox{\scriptsize s.t.} \ \varphi \in T} \ X_{(T,f)} \leq k-1$.
\item
If $\phi := \phi_1\wedge \phi_2$,
then 
$\prebatom_{\phi} := 
\prebatom_{\phi_1}\wedge\prebatom_{\phi_2}$.
\item
If $\phi := \phi_1\vee \phi_2$,
then 
$\prebatom_{\phi} :=  
\prebatom_{\phi_1}\vee\prebatom_{\phi_2}$.
\end{itemize}

Finally, the formula $\const(\vX)$ is to makes sure that
the solution to each variable $X_{T_i,f_j}$ is ``consistent''
to the intended meaning of the type $T_i$ and function $f_j$.
That is, for every types $S,T \in \cT_{\phi}$ the following holds.
\begin{itemize}\itemsep=0pt
\item
Every solution $\vM_T$ to the variables $\vX_T = (X_{T,f_1},\ldots,X_{T,f_m})$
corresponds to a $(D_T,\rev{D}_T)$-complete-regular digraph of size $\vM_T$,
where the matrices $D_T,\rev{D}_T \in \bbB^{\ell\times m}$ are as follows.
\begin{eqnarray*}
D_T & := &
\left(
\begin{array}{cccc}
f_1(T,R_1,T) & f_2(T,R_1,T) & \cdots & f_m(T,R_1,T) 
\\
f_1(T,R_2,T) & f_2(T,R_2,T) & \cdots & f_m(T,R_2,T) 
\\
\vdots & \vdots & \ddots & \vdots
\\
f_1(T,R_{\ell},T) & f_2(T,R_{\ell},T) & \cdots & f_m(T,R_{\ell},T) 
\end{array}
\right)
\end{eqnarray*}
and
\begin{eqnarray*}
\rev{D}_T & := &
\left(
\begin{array}{cccc}
f_1(T,\rev{R}_1,T) & f_2(T,\rev{R}_1,T) & \cdots & f_m(T,\rev{R}_1,T) 
\\
f_1(T,\rev{R}_2,T) & f_2(T,\rev{R}_2,T) & \cdots & f_m(T,\rev{R}_2,T) 
\\
\vdots & \vdots & \ddots & \vdots
\\
f_1(T,\rev{R}_{\ell},T) & f_2(T,\rev{R}_{\ell},T) & \cdots & f_m(T,\rev{R}_{\ell},T) 
\end{array}
\right)
\end{eqnarray*}
Notice that the matrix $D_{T}$ contains only
the information on the degree of $R_1,\ldots,R_{\ell}$,
while $\rev{D}_T$ the information on the degree of $\rev{R}_1,\ldots,\rev{R}_{\ell}$.
This is because the incoming $R_i$ edges to an element $v$
are precisely the outgoing $\rev{R}_i$ edges from $v$,
and vice versa, the incoming $\rev{R}_i$ edges from an element $v$
are precisely the outgoing $R_i$ edges to $v$.


\item
Every solution $\vM_S,\vM_T$ to the variables 
$\vX_S=(X_{S,f_1},\ldots,X_{S,f_m})$ and $\vX_T = (X_{T,f_1},\ldots,X_{T,f_m})$
corresponds to a $(D_{S\to T},\rev{D}_{S\to T})$-complete-biregular digraph 
of size $(\vM_S,\vM_T)$,
where the matrices $D_{S\to T},\rev{D}_{S\to T} \in \bbB^{\ell\times m}$ are as follows.
\begin{eqnarray*}
D_{S\to T} & := &
\left(
\begin{array}{cccc}
f_1(S,R_1,T) & f_2(S,R_1,T) & \cdots & f_m(S,R_1,T) 
\\
f_1(S,R_2,T) & f_2(S,R_2,T) & \cdots & f_m(S,R_2,T) 
\\
\vdots & \vdots & \ddots & \vdots
\\
f_1(S,R_{\ell},T) & f_2(S,R_{\ell},T) & \cdots & f_m(S,R_{\ell},T) 
\\
f_1(S,\rev{R}_1,T) & f_2(S,\rev{R}_1,T) & \cdots & f_m(S,\rev{R}_1,T) 
\\
f_1(S,\rev{R}_2,T) & f_2(S,\rev{R}_2,T) & \cdots & f_m(S,\rev{R}_2,T) 
\\
\vdots & \vdots & \ddots & \vdots
\\
f_1(S,\rev{R}_{\ell},T) & f_2(S,\rev{R}_{\ell},T) & \cdots & f_m(S,\rev{R}_{\ell},T) 
\end{array}
\right)
\end{eqnarray*}
and
\begin{eqnarray*}
\rev{D}_{S\to T} & := &
\left(
\begin{array}{cccc}
f_1(T,\rev{R}_1,S) & f_2(T,\rev{R}_1,S) & \cdots & f_m(T,\rev{R}_1,S) 
\\
f_1(T,\rev{R}_2,S) & f_2(T,\rev{R}_2,S) & \cdots & f_m(T,\rev{R}_2,S) 
\\
\vdots & \vdots & \ddots & \vdots
\\
f_1(T,\rev{R}_{\ell},S) & f_2(T,\rev{R}_{\ell},S) & \cdots & f_m(T,\rev{R}_{\ell},S) 
\\
f_1(T,R_1,S) & f_2(T,R_1,S) & \cdots & f_m(T,R_1,S) 
\\
f_1(T,R_2,S) & f_2(T,R_2,S) & \cdots & f_m(T,R_2,S) 
\\
\vdots & \vdots & \ddots & \vdots
\\
f_1(T,R_{\ell},S) & f_2(T,R_{\ell},S) & \cdots & f_m(T,R_{\ell},S) 
\end{array}
\right)
\end{eqnarray*}
Notice that in the matrix $D_{S\to T}$
the first $\ell$ rows contains 
the information on the degree of $R_1,\ldots,R_{\ell}$,
and the last $\ell$ rows 
the information on the degree of $\rev{R}_1,\ldots,\rev{R}_{\ell}$
{\em from the type $S$ to the type $T$};
while in the matrix $\rev{D}_{S\to T}$ it is the opposite
and the direction is {\em from the type $T$ to the type $S$}.
Similar as in the $D_{T},\rev{D}_t$ case above,
this is because the incoming $R_i$ edges to an element $v$
are precisely the outgoing $\rev{R}_i$ edges from $v$,
and vice versa, the outgoing $R_i$ edges from an element $v$
are precisely the incoming $\rev{R}_i$ edges to $v$.
\end{itemize}
Now the formula $\const(\vX)$ is simply the conjunction:
\begin{eqnarray*}
\mbox{\hspace{1.0 cm}}\const(\vX) & := &
\bigwedge_{1 \leq i \leq n} \diregc_{D_{T_i},\srev{D}_{T_i}}(\vX_T)
\\
& & 
\wedge \;
\bigwedge_{1 \leq j < i \leq n}
\
\biregc_{D_{T_i\to T_j},\srev{D}_{T_i\to T_j}}(\vX_{T_i},\vX_{T_j})
\end{eqnarray*}
where 
$\vX_{T_i} = (X_{(T_i,f_1)},\ldots,X_{(T_i,f_m)})$
is the vector of variables associated with the type $T_i$.

We are going to show that $\preb_{\phi}$ defines precisely the spectrum of $\phi$,
as stated in the claim below.
Abusing the notation, we let $\preb_{\phi}$ itself to denote
the set $\{n \mid \preb_{\phi}(n)\ \mbox{holds}\}$.
Recall also that as defined in Section~\ref{s: qmlc},
the spectrum of a $\QMLC$ sentence $\phi$ is restricted
to the complete structures.

\begin{claim}
For every $\QMLC$ sentence $\phi$,
$\preb_{\phi}=\spec(\phi)$, where $\preb_{\phi}(x)$ is
the formula
\begin{eqnarray*}
& \exists X_{(T_1,f_1)}  \cdots \exists X_{(T_n,f_m)} &
\Bigg(x = \sum_{1\leq i \leq n} \ \ \sum_{1 \leq j \leq m} \ X_{T_i,f_j}\Bigg)
\; \wedge \; \prebatom_{\phi}(\vX)
\; \wedge\; 
\const(\vX)
\end{eqnarray*}
and $\vX = (X_{(T_1,f_1)},\ldots,X_{(T_n,f_m)})$ is
the vector of all the variables $X_{(T,f)}$'s.
\end{claim}

The proof is by induction on $\phi$.
The base case is when $\phi$ is a basic $\QMLC$ sentence or the negation of a basic $\QMLC$ sentence.
We consider first the case when $\phi$ is a basic $\QMLC$ formula
of the form $\exists^k \ \varphi$, where $\varphi \in \MLC$.
In this case $\prebatom_{\phi}(\vX)$ is
$\sum_{(T,f) \ \mbox{\scriptsize s.t.} \ \varphi \in T} \ X_{(T,f)} \geq k$.
 
We first show the direction $\preb_{\phi}\subseteq \spec(\phi)$.
Let $N \in \preb_{\phi}$.
Let $\vM = (M_{T_1,f_1},\ldots,M_{T_n,f_m})$ be the
witnesses to $\vX$ such that $\preb_{\phi}(N)$ holds.
In the following we are going to write
$\vM_T$ to denote $(M_{T,f_1},\ldots,M_{T,f_m})$
for every type $T\in \cT_{\phi}$.

Since $x = \sum_{(T,f)} X_{T,f}$,
we have $N = \sum_{(T,f)} M_{(T,f)}$.
We take a set $V$ of $N$ vertices
and we partition $V = V_{(T_1,f_1)}\cup\cdots\cup V_{(T_n,f_m)}$
such that $|V_{(T,f)}| = M_{(T,f)}$ for each $T\in \cT_{\phi}$ and $f\in \cF$.
We denote by $V_T = V_{(T,f_1)}\cup\cdots\cup V_{(T,f_m)}$
for each $T \in \cT_{\phi}$.

Since $\const(\vM)$ holds,
by Theorems~\ref{t: main regular-digraph},
for each $T \in \cT_{\phi}$,
there exists a $(D_T,\rev{D}_T)$-complete-regular digraph 
$G_T = (V_T,R_{T,1},\ldots,R_{T,\ell})$ of size $\vM_T$,
with $V_T = V_{(T,f_1)}\cup\cdots\cup V_{(T,f_m)}$
be the witness of the $(D_T,\rev{D}_T)$-regularity.
This means that for every vertex $v \in V_{T,f_i}$,
for every $R \in \{R_1,\ldots,R_{\ell}\}$,
\begin{itemize}\itemsep=0pt
\item
$\outdeg_R(v)$ in the graph $G_T$ is $f_i(T,R,T)$;
\item
$\indeg_R(v)$ in the graph $G_T$ is $f_i(T,\rev{R},T)$.
\end{itemize}
Now let $\tilde{G}_T = (V_T, R_{T,1},\ldots,R_{T,\ell},
\rev{R}_{T,1},\ldots,\rev{R}_{T,\ell})$
be the graph obtained by taking $\rev{R}_i$
as the inverse of $R_i$.
Then for each vertex $v \in V_T$,
\begin{itemize}\itemsep=0pt
\item
$\outdeg_R(v) = \indeg_{\srev{R}}(v)$ in the graph $\tilde{G}_T$;
\item
$\indeg_R(v) = \outdeg_{\srev{R}}(v)$ in the graph $\tilde{G}_T$.
\end{itemize}

Similarly, by Theorem~\ref{t: main biregular}
for each $T_i,T_j\in \cT_{\phi}$, where $j \leq i-1$,
there exists a $(D_{T_i\to T_j},\rev{D}_{T_i\to T_j})$-biregular-complete graph 
\[
G_{T_i,T_j} = (V_{T_i},V_{T_j},R_{T_i,T_j,1},\ldots,R_{T_i,T_j,\ell},
\rev{R}_{T_i,T_j,1},\ldots,\rev{R}_{T_i,T_j,\ell})
\]
of size $(\vM_S,\vM_T)$,
with $V_{T_i} = V_{(T_i,f_1)}\cup\cdots\cup V_{(T_i,f_m)}$
and $V_{T_j} = V_{(T_j,f_1)}\cup\cdots\cup V_{(T_j,f_m)}$
be the witness of the $(D_{T_i\to T_j},\rev{D}_{T_i\to T_j})$-biregularity.
This means that 
for every $R \in \{R_1,\ldots,R_{\ell},\rev{R}_1,\ldots,\rev{R}_{\ell}\}$,
\begin{itemize}\itemsep=0pt
\item
for every vertex $v \in V_{T_i,f}$,
$\outdeg_R(v)$ in the graph $G_{T_i,T_j}$ is $f(T_i,R,T_j)$;
\item
for every vertex $v \in V_{T_j,f}$,
$\outdeg_R(v)$ in the graph $G_{T_i,T_j}$ is $f(T_i,\rev{R},T_j)$.
\end{itemize}
We put the orientation in every the edges in the graph $G_{T_i,T_j}$
going from $V_{T_i}$ to $V_{T_j}$.
Now let $\tilde{G}_{T_i,T_j}$
be the graph obtained by adding $(u,v)$ into $\rev{R}$ in the graph $G_{T_i,T_j}$ 
whenever $(v,u)$ is an $R$-edge in $G_{T_i,T_j}$.

Hence, we have for each vertex $v \in V_{T_i} \cup V_{T_j}$, for each $R \in \Rels$
\begin{itemize}\itemsep=0pt
\item
$\outdeg_R(v) = \indeg_{\srev{R}}(v)$ in the graph $\tilde{G}_{T_i,T_j}$;
\item
$\indeg_R(v) = \outdeg_{\srev{R}}(v)$ in the graph $\tilde{G}_{T_i,T_j}$.
\end{itemize}
Let $G= (V, R_1,\ldots,R_{\ell},\rev{R}_1,\ldots,\rev{R}_{\ell})$ 
be the combination of all the graphs $\tilde{G}_{T_i}$'s and $\tilde{G}_{T_i,T_j}$'s.
Formally,
\begin{eqnarray*}
V & = & \bigcup_{T} V(\tilde{G}_t)
\\
R & = & \bigcup_{T} R(\tilde{G}_T) 
\ \ \cup \ \ \bigcup_{T_i,T_j} R(\tilde{G}_{T_i,T_j}) 
\qquad \mbox{for each} \ R \in \{R_1,\ldots,R_{\ell},\rev{R}_1,\ldots,\rev{R}_{\ell}\}
\end{eqnarray*}
Moreover, we also label each vertex $v \in V$ with a subset of $\cS$ as follows.
For each $T \in \cT_{\phi}$,
for each $v \in V_T$,
we ``declare'' that $v$ is labeled with a unary predicate $P\in \cS$
if and only if $P \in T$. 

We claim that $G \models \phi$.
For that it is sufficient to show that for
each $T \in \cT_{\phi}$,
for each $v \in V_T$,
$\tp_G(v) = T$.
The proof is divided into three cases.
\begin{itemize}\itemsep=0pt
\item
For each unary predicate $P \in \cP$,
it is by our labelling of the vertices of $G$ that
$P(v)$ holds in $G$ if and only if $P \in T$.
\item
For each $\Diamond_R^l \ \mu \in T$, we have
\begin{eqnarray*}
\sum_{T' \ \mbox{\scriptsize s.t.} \ T' \ni \mu} \ f(T,R,T')   &\geq &  l
\end{eqnarray*}
number of outgoing $R$-edges from $v$.
Since every function $f \in \cF$ is consistent,
$\Diamond_R^l \ \mu \in \tp_G(v)$.

\item
Similary, for each $\Diamond_R^l \ \mu \notin T$, 
and hence $\neg (\Diamond_R^l \ \mu) \in T$, 
we have
\begin{eqnarray*}
\sum_{T' \ \mbox{\scriptsize s.t.} \ T' \ni \mu} \ f(T,R,T') & \leq &  l-1
\end{eqnarray*}
number of outgoing $R$-edges from $v$.
Since every function $f \in \cF$ is consistent,
$\Diamond_R^l \ \mu \notin \tp_G(v)$.
\end{itemize}
Therefore the graph $G \models \phi$,
and hence $N \in \spectrum(\phi)$.

Now we prove the direction $\preb_{\phi}\supseteq \spec(\phi)$.
Suppose $\stra \models \phi$ and $\stra$ is of size $N$.
Let $\vM = (M_{(T_1,f_1)},\ldots,M_{(T_n,f_m,)})$
where $M_{(T,f)}$ be the number of elements of type $T$
from which there exist $f(T,R,S)$ number of outgoing $R$-edges
towards the elements of type $S$.
Take each $M_{(T,f)}$ to be the witness for $X_{(T,f)}$
for each $T \in \cT_{\phi}$ and $f\in \cF$.
It immediately follows from 
Theorems~\ref{t: main biregular} and~\ref{t: main regular-digraph}
that $\const(N,\vM)$ holds.
Moreover, $\prebatom_{\phi}(\vM)$ holds, since $\stra \models \phi$.
This completes the proof of $\preb_{\phi}=\spec(\phi)$,
when $\phi$ is a basic $\QMLC$ sentence.

When $\phi \in \QMLC$ is the negation of a basic $\QMLC$ sentence,
say $\neg \exists^k \varphi$,
the formula $\prebatom_{\phi}$ is
\begin{eqnarray*}
\sum_{(T,f) \ \mbox{\scriptsize s.t.} \ \varphi \in T} \ X_{(T,f)} & \leq & k-1
\end{eqnarray*}
which is the negation of 
\begin{eqnarray*}
\sum_{(T,f) \ \mbox{\scriptsize s.t.} \ \varphi \in T} \ X_{(T,f)} & \geq & k
\end{eqnarray*}
Then $\preb_{\phi} = \spec(\phi)$ follows immediately from above.

The correctness for the case 
when $\phi$ is $\phi_1\wedge \phi_2$ or $\phi_1\vee \phi_2$,
can be established via straightforward inductive argument.
This completes our proof that $\preb_{\phi}=\spec(\phi)$,
and hence, Theorem~\ref{t: presburger for c2}.

\section{Proof of Theorem~\ref{t: main biregular}}
\label{s: proof1}

The proof of Theorem~\ref{t: main biregular} is rather long.
As a warm-up, we prove the following easy Proposition~\ref{p: d-biregular} first.

\begin{proposition}
\label{p: d-biregular}
Let $c,d\geq 0$. For every $M,N \in \bbN$,
the following holds.
\begin{enumerate}[(a)]\itemsep=0pt
\item
There exists a $(c,d)$-biregular graph of size $(M,N)$ 
if and only if $N\geq c$, $M\geq d$ and $M\cdot c = N\cdot d$.
\item
There exists a $(c,{\ugeq d})$-biregular graph of size $(M,N)$  
if and only if $M\geq d$, $N \geq c$ and $Mc \geq Nd$.
\item
There exists a $({\ugeq c},{\ugeq d})$-biregular graph 
of size $(M,N)$ if and only if $M\geq d$, $N \geq c$.
\end{enumerate}
\end{proposition}
\begin{proof}
Let $c,d\geq 0$, and let $M,N \in \bbN$.
We first prove part~(a).
The ``only if'' direction follows from the fact that
in $(c,d)$-biregular graph
the number of edges is precisely $M c=N d$.
That $M \geq d$ and $N \geq c$ is straightforward.

The ``if'' direction is as follows.
Let $K = Mc=Nd$.
Suppose also $M\geq d$ and $N \geq c$.

First, we construct the following graph.


{\footnotesize
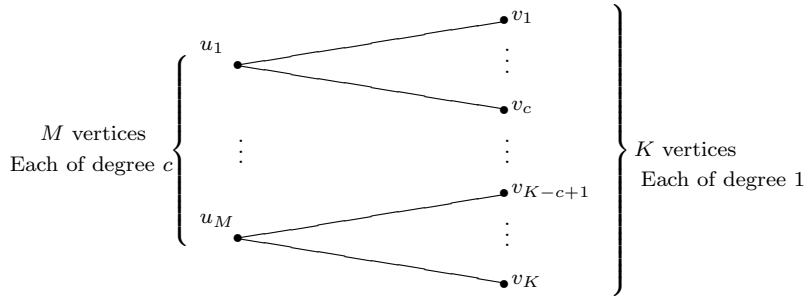
\begin{picture}(200,40)(-110,0)

\put(-80,18){
$\left\{\begin{array}{c}
\\\\\\\\\\\\\\\\  \end{array}\right.$}
\put(-96,20){$M$ vertices}
\put(-100,16){Each of degree $c$}

\put(-75,32){$u_{1}$}
\put(-70,30){\circle*{1}}
\put(-70,30){\line(6,1){35}}
\put(-35,36){\circle*{1}}
\put(-34,36){$v_1$}
\put(-35,29){$\vdots$}
\put(-70,30){\line(6,-1){35}}
\put(-35,24){\circle*{1}}
\put(-34,24){$v_{c}$}

\put(-70,17){$\vdots$}
\put(-35,17){$\vdots$}

\put(-75,9){$u_{M}$}
\put(-70,7){\circle*{1}}
\put(-70,7){\line(6,1){35}}
\put(-35,13){\circle*{1}}
\put(-34,13){$v_{K-c+1}$}
\put(-35,6){$\vdots$}
\put(-70,7){\line(6,-1){35}}
\put(-35,1){\circle*{1}}
\put(-34,1){$v_K$}

\put(-25,18){$\left.\begin{array}{c}
\\\\\\\\\\\\\\\\\\\\\\\\ \end{array}\right\} K$ vertices}
\put(-17,14){Each of degree $1$}

\end{picture}
}


On the left side, 
we have $M$ vertices, and each has degree $c$.
On the right side,
we have $K=Nd$ vertices, and each has degree $1$.
We are going to merge every $d$ vertices on the right side into one vertex of degree $d$.
The merging is as follows.
We merge every $d$ vertices $v_i,v_{i+N},\ldots,v_{i+(d-1)N}$
into one node for every $i=1,\ldots,N$.
Since $K=Nd$, it is possible to do such merging.
Moreover, $N \geq c$, hence 
we do not have multiple edges between two vertices.
Thus, we obtain the desired $(c,d)$-biregular graph of size $(M,N)$.\footnote{Note that 
since $Mc=Nd$, $N\geq c$ already implies $M\geq d$.
That is why it is not necessary to use the fact that $M\geq d$.
Of course, by symmetry, we can first build a bipartite graph
in which the left side has $N$ vertices of degree $d$ and the right side has
$K=Mc$ vertices of degree $1$. Then to build the desired $(c,d)$-biregular graph,
we make use of the fact $M \geq d$.}

Now we consider part~(b).
The ``only if'' direction follows from the fact that in $(c,{\ugeq d})$-biregular graph
the number of edges is precisely $Mc$, 
which should be greater than $Nd$.
That $M\geq d$ and $N\geq c$ is straightforward.

For the ``if'' direction,
the proof is almost the same as above.
Suppose $M\geq d$, $N \geq c$ and $M c \geq Nd$.
Let $K = Mc$.
We first construct the bipartite graph, in which
on the left side, 
we have $M$ vertices, each of which has degree $c$;
and on the right side
we have $K=Mc$ vertices, each of which has degree $1$.

For every $i \in \{1,\ldots,N\}$,
we set the set $I_i \subseteq \{1,\ldots,K\}$ as follows. 
\[
I_i := \{i + k N \mid 1\leq i+kN \leq K\}.
\]
Since $K=Mc \geq Nd$,
the cardinality $|I_i| \geq d$ for every $1\leq i \leq N$.
Now for every $1\leq i \leq N$,
we merge the vertices $\{v_j \mid j \in I_i\}$ into one vertex. 
Hence, we obtain the desired $(c,{\ugeq d})$-biregular graph of size $(M,N)$.

Now we prove part~(c).
The ``only if'' direction is straightforward.
The ``if'' direction is as follows.
Suppose $M \geq d$ and $N \geq c$.
There are two cases: either $Mc \geq N d$,
or $Mc < Nd$.
In the former case, we construct 
a $(c,{\ugeq d})$-biregular graph of size $(M,N)$,
while in the latter case, we construct 
a $({\ugeq c},d)$-biregular graph of size $(M,N)$.
In either case, we obtain a $({\ugeq c},{\ugeq d})$-biregular graph
of size $(M,N)$.
This completes our proof of Proposition~\ref{p: d-biregular}.
\end{proof}

The proof of Theorem~\ref{t: main biregular} is a 
generalisation of Proposition~\ref{p: d-biregular} above.
It is divided into five successive steps
presented in Subsections~\ref{ss: 1-type biregular}--~\ref{ss: formula l-type complete biregular}.
\begin{itemize}\itemsep=0pt
\item
Subsection~\ref{ss: 1-type biregular}.

It contains the generalisation of Proposition~\ref{p: d-biregular} 
to the case of $(\vc,\vd)$-biregular,
where $\vc$ and $\vd$ are vectors over $\bbN$.
(That is, we consider the 1-type biregular graphs.)

\item
Subsection~\ref{ss: l-type biregular}.

In this subsection
we will use the Presburger characterisation for $(\vc,\vd)$-biregular graphs
to obtain the same characterisation for $\ell$-type $(C,D)$-biregular graphs,
where $C,D$ are matrices over $\bbN$. 
\item
Subsection~\ref{ss: ugeq biregular}.

In this subsection we obtain the characterisation for $(C,D)$-biregular graphs
when $C,D$ contain elements from ${\ugeq\bbN}$
assuming that the number of vertices whose degrees specified with ${\ugeq d}$ is ``big enough.''  
It is obtained by using the characterisation in the previous Subsection~\ref{ss: l-type biregular}.

\item
Subsection~\ref{ss: l-type ugeq biregular}.

This subsection is the generalisation of Subsection~\ref{ss: ugeq biregular},
where the graphs may contain a ``small'' number of vertices whose degree is specified with ${\ugeq d}$.
The idea is to encode directly those vertices in the Presburger formula.
This is presented formally by our notion of {\em partial graphs}.

\item
Subsection~\ref{ss: formula l-type biregular}.

In this subsection we present the construction of the 
formula required in Theorem~\ref{t: l-type biregular}.
It is built from the formula presented in the Subsection~\ref{ss: l-type ugeq biregular}.

\item
Subsection~\ref{ss: formula l-type complete biregular}.

Finally in this subsection we present the construction of the 
formula required in Theorem~\ref{t: main biregular},
which is built from the from the formula in Subsection~\ref{ss: formula l-type biregular}.
\end{itemize}

In the following we write $\vone$ 
to denote the vector $(1,\ldots,1)\in\bbN^m$, for an appropriate $m \geq 1$.
That is, $\vone$ is a vector whose components are all one.
For two vectors $\vc = (c_1,\ldots,c_m) \in \bbN^m$ and $\vd = (d_1,\ldots,d_m) \in \bbN^m$,
the dot product between $\vc$ and $\vd$ is $\vc\cdot\vd = c_1d_1 + \cdots +c_md_m$.

\subsection{When $C=\vc \in \bbN^{1\times m}$
and $D=\vd \in \bbN^{1\times n}$}
\label{ss: 1-type biregular}

In this subsection we consider the case when $C$ and $D$ consist of only one vector each.
In this case, we are going to write $(\vc,\vd)$-biregular graph,
where $\vc$ and $\vd$ are the only vectors of $C$ and $D$, respectively.

\begin{lemma}
\label{l: vd-biregular}
Let $\vc\in \bbN^m$ and $\vd \in \bbN^n$
and both do not contain zero entry.
For each $\vM\in \bbN^m$ and $\vN \in \bbN^n$ such that
$\vM\cdot\vone+\vN\cdot\vone  \geq 2(\vc\cdot\vone)(\vd\cdot\vone)+3$,
the following holds.
There exists a $(\vc,\vd)$-biregular graph of size $(\vM,\vN)$
if and only if $\vM \cdot \vc = \vN \cdot \vd$.
\end{lemma}
\begin{proof}
Let $\vc \in \bbN^m$ and $\vd \in \bbN^n$ and
both do not contain zero entry.
Let $\vM\in \bbN^m$, $\vN \in \bbN^n$ such that
$\vM\cdot\vone+\vN\cdot\vone \geq 2(\vc\cdot\vone)(\vd\cdot\vone) +3$.

The ``only if'' direction is straightforward.
If $G$ is a $(\vc,\vd)$-biregular graph of size $(\vM,\vN)$,
then the number of edges in $G$ is precisely $\vM\cdot\vc=\vN\cdot\vd$.

Now we prove the ``if'' part.
Suppose $\vM\in \bbN^m$, $\vN \in \bbN^n$ such that
$\vM\cdot\vc = \vN\cdot\vd$.
Let $\vc = (c_{1},\ldots,c_{m})$ and $\vd = (d_{1},\ldots,d_{n})$,
and $\vM=(M_{1},\ldots,M_{m})$ and $\vN = (N_{1},\ldots,N_{n})$.

We are going to construct a $(\vc,\vd)$-biregular graph of size $(\vM,\vN)$.
We first construct a preliminary bipartite graph $G$ pictured in Figure~\ref{fig:preliminary}.
That is, the left side has $\vM\cdot\vone$ vertices, and
there are $M_{1}$ vertices of degree $c_{1}$,
$M_{2}$ nodes of degree $c_{2}$, etc.
The right side has $\vM\cdot\vc$ number of vertices, each of degree one.

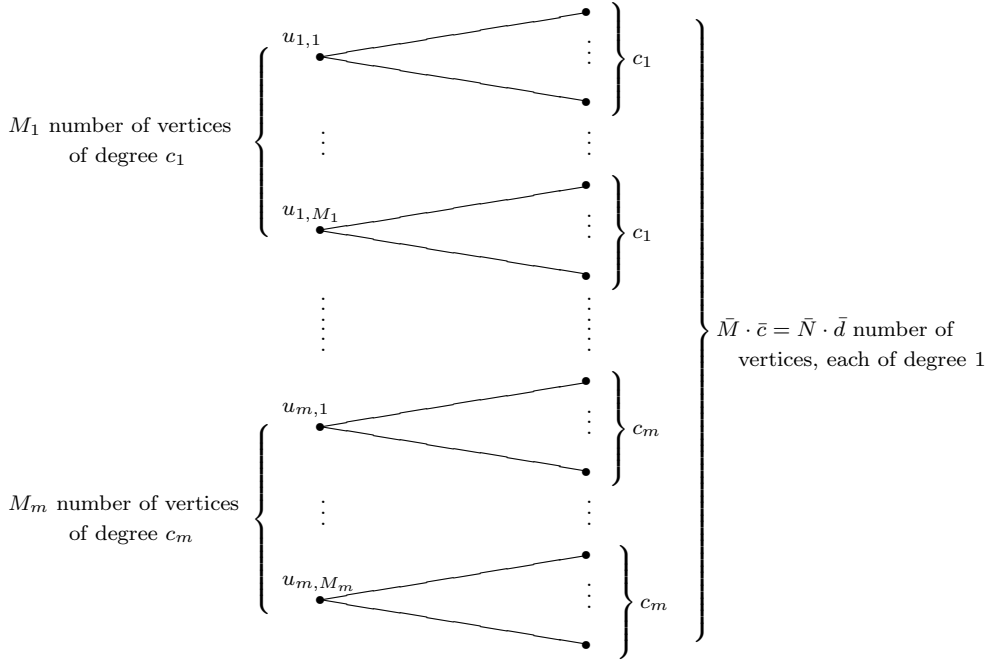
\begin{figure}[t]
\begin{picture}(200,100)(-112,-55)

{\footnotesize

\put(-80,18){
$\left\{\begin{array}{c}
\\\\\\\\\\\\\\\\  \end{array}\right.$}
\put(-111,20){$M_{1}$ number of vertices}
\put(-103,16){of degree $c_{1}$}

\put(-75,32){$u_{1,1}$}
\put(-70,30){\circle*{1}}
\put(-70,30){\line(6,1){35}}
\put(-35,36){\circle*{1}}
\put(-35,29){$\vdots$}
\put(-70,30){\line(6,-1){35}}
\put(-35,24){\circle*{1}}
\put(-36,29){$\left.\begin{array}{c}\\ \\ \\ \\ \\ \end{array}\right\} c_{1}$}

\put(-70,17){$\vdots$}
\put(-35,17){$\vdots$}

\put(-75,9){$u_{1,M_{1}}$}
\put(-70,7){\circle*{1}}
\put(-70,7){\line(6,1){35}}
\put(-35,13){\circle*{1}}
\put(-35,6){$\vdots$}
\put(-70,7){\line(6,-1){35}}
\put(-35,1){\circle*{1}}
\put(-36,6){$\left.\begin{array}{c}\\ \\ \\ \\ \\ \end{array}\right\} c_{1}$}

\put(-70,-5){$\vdots$}
\put(-70,-9){$\vdots$}
\put(-35,-5){$\vdots$}
\put(-35,-9){$\vdots$}

\put(-80,-32){
$\left\{\begin{array}{c}
\\\\\\\\\\\\\\\\  \end{array}\right.$}
\put(-111,-30){$M_{m}$ number of vertices}
\put(-103,-34){of degree $c_{m}$}

\put(-75,-17){$u_{m,1}$}
\put(-70,-19){\circle*{1}}
\put(-70,-19){\line(6,1){35}}
\put(-35,-13){\circle*{1}}
\put(-35,-20){$\vdots$}
\put(-70,-19){\line(6,-1){35}}
\put(-35,-25){\circle*{1}}
\put(-36,-20){$\left.\begin{array}{c}\\ \\ \\ \\ \\ \end{array}\right\} c_{m}$}

\put(-70,-32){$\vdots$}
\put(-35,-32){$\vdots$}

\put(-75,-40){$u_{m,M_{m}}$}
\put(-70,-42){\circle*{1}}
\put(-70,-42){\line(6,1){35}}
\put(-35,-36){\circle*{1}}
\put(-35,-43){$\vdots$}
\put(-70,-42){\line(6,-1){35}}
\put(-35,-48){\circle*{1}}
\put(-35,-43){$\left.\begin{array}{c}\\ \\ \\ \\ \\ \end{array}\right\} c_{m}$}

\put(-25,-7){$\left.\begin{array}{c}
\\\\\\\\\\\\\\\\\\\\\\\\\\\\\\\\\\\\\\ \\ \\ \\ \\ \\ \\ \end{array}\right\} 
\vM\cdot\vc=\vN\cdot\vd$ number of}
\put(-15,-11){vertices, each of degree $1$}
}
\end{picture}
\caption{The preliminary graph constructed in the proof of Lemma~\ref{l: vd-biregular}.}
\label{fig:preliminary}
\end{figure}

We are going to do some merging of the vertices on the right side
so that there are exactly $N_{1}$ vertices of degree $d_{1}$,
$N_{2}$ vertices of degree $d_{2}$, etc.
We do the following.
We ``group'' the vertices on the right side into $V_1,\ldots,V_n$
where $V_1$ has $N_{1}d_{1}$ vertices,
$V_2$ has $N_{2}d_{2}$ vertices, etc.
Such grouping is possible because $\vM\cdot\vc = \vN\cdot \vd$.

For each $i\in \{1,\ldots,n\}$, we do the following.
We merge $d_{i}$ vertices in $V_i$ into one vertex,
so that each vertex in $V_i$ has degree $d_{i}$.
Let $V_i = \{v_{i,1},\ldots,v_{i,K_i}\}$ where $K_i= N_{i}d_{i}$.
We merge the vertices 
$v_1,v_{N_{i}+1},v_{2N_{i}+1},\ldots,v_{(d_{i}-1)N_{i}+1}$ into one vertex;
the vertices $v_2,v_{N_{i}+2},v_{2N_{i}+2},\ldots,v_{(d_{i}-1)N_{i}+2}$
into one vertex; and so on.

After such merging, each vertex in $V_i$ has degree $d_{i}$.
However, it is possible that after we do the merging,
we have ``parallel'' edges, i.e. more than one edges between two vertices.
(See the left side of the illustration below.)
We are going to ``remove'' such parallel edges one by one
until there are no more parallel edges.

Suppose we have parallel edges between the vertices $u$ and $v$.
We pick an edge $(u',v')$ such that $u'$ is not adjacent to $v$ and $v'$ is not adjacent to $u$.
(See the left side of the illustration below.)

\noindent
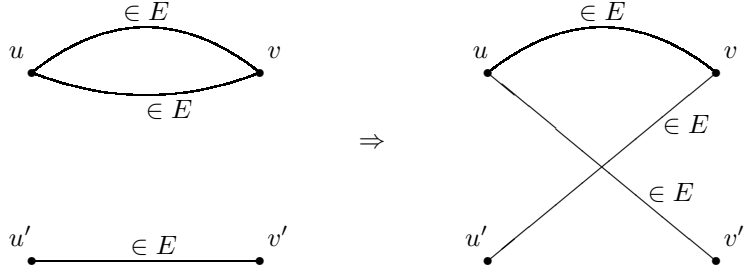
\begin{picture}(200,42)(-110,-21)

\put(-93,12){$u$}
\put(-90,10){\circle*{1}}
\qbezier(-90,10)(-75,22)(-60,10)
\put(-78,17){$\in E$}
\qbezier(-90,10)(-75,4)(-60,10)
\put(-75,4){$\in E$}
\put(-59,12){$v$}
\put(-60,10){\circle*{1}}

\put(-93,-13){$u'$}
\put(-90,-15){\circle*{1}}
\put(-77,-14){$\in E$}
\put(-90,-15){\line(1,0){30}}
\put(-59,-13){$v'$}
\put(-60,-15){\circle*{1}}

\put(-47,0){$\Rightarrow$}

\put(1,12){$v$}
\put(0,10){\circle*{1}}
\put(-18,17){$\in E$}
\qbezier(0,10)(-15,22)(-30,10)
\put(-32,12){$u$}
\put(-30,10){\circle*{1}}
\put(0,10){\line(-6,-5){30}}
\put(-7,2){$\in E$}

\put(1,-13){$v'$}
\put(0,-15){\circle*{1}}
\put(-9,-7){$\in E$}
\put(0,-15){\line(-6,5){30}}
\put(-33,-13){$u'$}
\put(-30,-15){\circle*{1}}

\end{picture}

\noindent
Such an edge $(u',v')$ exists since
the number of vertices reachable in distance $2$ from
the vertices $u$ and $v$ is $\leq 2(\vc\cdot\vone)(\vd\cdot\vone)+2$
and the number of vertices is 
$\vM\cdot\vone+\vN\cdot\vone \geq 2(\vc\cdot\vone)(\vd\cdot\vone)+3$
and the fact that none of the vertices are of zero degree.
(Here we make use of the fact that neither $\vc$ nor $\vd$ contain zero entry.)

Now we delete the edges $(u',v')$ and one of the parallel edge $(u,v)$,
replace it with the edges $(u,v')$ and $(u',v)$,
as illustrated on the right side of the illustration above.
We perform such operation until there are no more parallel edges.
This completes the proof of Lemma~\ref{l: vd-biregular}.
\end{proof}

The following theorem is a straightforward application of Lemma~\ref{l: vd-biregular}.
\begin{theorem}
\label{t: vd-biregular}
For every $\vc\in \bbN^m$ and $\vd \in \bbN^n$,
there exists a Presburger formula $\bireg_{(\vc,\vd)}(\vX,\vY)$,
where $\vX = (X_1,\ldots,X_m)$ and $\vY = (Y_1,\ldots,Y_n)$
such that the following holds.
There exists a $(\vc,\vd)$-biregular graph
of size $(\vM,\vN)$ if and only if
the sentence $\bireg_{(\vc,\vd)}(\vM,\vN)$ holds.
\end{theorem}
\begin{proof}
The proof is a direct application of Lemma~\ref{l: vd-biregular}.
We assume that all the entries in $\vc$ and $\vd$ are not zero.
Otherwise, we do the following.
Suppose $\vc=(c_1,\ldots,c_m)$, $\vd = (d_1,\ldots,d_n)$ and 
let $I = \{i \mid c_i =0\}$ and $J = \{j \mid d_j=0\}$.
We define
\begin{eqnarray*}
\bireg_{\vc,\vd}(\vX,\vY) &:= & 
\bigwedge_{i \in I}\quad X_{i}\geq 0 \ \wedge \ 
\bigwedge_{j \in J}\quad Y_{j}\geq 0 \ \wedge \ 
\bireg_{\vc',\vd'}(\vX',\vY'),
\end{eqnarray*}
where $\vc'$ and $\vX'$ are the vectors $\vc$ and $\vX$
without the entries in $I$, respectively,
and $\vd'$ and $\vY'$ are the vectors $\vd$ and $\vY$ 
without the entries in $J$, respectively.

For $\vc\in \bbN^m$ and $\vd \in \bbN^n$
which do not contain zero entry,
we define the following set $H$.
\[
H :=
\left\{
\begin{array}{ll}
(\vM,\vN) &
\left|
\begin{array}{l} 
\vM\cdot\vone+\vN\cdot\vone  \leq 2(\vc\cdot\vone)(\vd\cdot\vone)+2 \ \mbox{and}
\\ 
\mbox{there exists a} \ (\vc,\vd)\mbox{-biregular graph of size} \ (\vM,\vN) 
\end{array}
\right\}
\end{array}
\right.
\]
Such set can be computed greedily
since the number of $(\vM,\vN)$ 
such that $\vM\cdot\vone+\vN\cdot\vone  \leq 2(\vc\cdot\vone)(\vd\cdot\vone)+2$
is bounded.

Now we define the formula $\bireg_{\vc,\vd}(\vX,\vY)$ as follows.
\[
\Big(\vM\cdot\vone+\vN\cdot\vone  \geq 2(\vc\cdot\vone)(\vd\cdot\vone)+3 
\ \wedge \ (\vX\cdot\vc = \vY\cdot\vd)\Big)
\vee \:\bigvee_{(\vM,\vN) \in H} \vX = \vM \ \wedge \ \vY = \vN
\]
The formula is a Presburger formula since $\vc$ and $\vd$ are constants.
Since $\vc,\vd$ do not contain zero entry,
we can apply Lemma~\ref{l: vd-biregular} to obtain the correctness of $\bireg_{\vc,\vd}(\vX,\vY)$.
This completes our proof of Theorem~\ref{t: vd-biregular}.
\end{proof}

\subsection{When $C \in \bbN^{\ell\times m}$
and $D \in \bbN^{\ell\times n}$}
\label{ss: l-type biregular}

Theorem~\ref{t: l-type vd-biregular} below
is the generalisation of
Theorem~\ref{t: vd-biregular}
to the case where $\ell\geq 1$.
Recall that for a matrix $C \in \bbN^{\ell\times m}$,
we write $C\cdot\vone$ to denote the sum of all the entries in $C$.

\begin{theorem}
\label{t: l-type vd-biregular}
For every $C \in \bbN^{\ell\times m}$ and $D \in \bbN^{\ell\times n}$,
there exists a Presburger formula $\bireg_{C,D}(\vX,\vY)$,
where $\vX = (X_1,\ldots,X_m)$ and $\vY = (Y_1,\ldots,Y_n)$
such that the following holds.
There exists a $(C,D)$-biregular $\ell$-type graph
of size $(\vM,\vN)$ if and only if
the sentence $\bireg_{C,D}(\vM,\vN)$ holds.
\end{theorem}
\begin{proof}
Let $C \in \bbN^{\ell\times m}$ 
and $D \in \bbN^{\ell\times n}$ be the given matrices.
For simplicity, we assume that both $C$ and $D$ do not contain zero column.
If column $i$ in matrix $C$ (or, $D$, respectively) is zero column, 
then we add the constraint $X_i\geq 0$ (or, $Y_i\geq 0$, respectively)
and ignore that column.

Let $\vc_1,\ldots,\vc_{\ell}$ and $\vd_1,\ldots,\vd_{\ell}$
be the row vectors of $C$ and $D$, respectively.
For a vector $\vt = (t_1,\ldots,t_m) \in \bbN^m$,
we define the characteristic vector of $\vt$
as $\chi(\vt) := (b_1,\ldots,b_m) \in \{0,1\}^m$ 
where $b_i = 0$ if $t_i = 0$, and $b_i =1$ if $t_i \neq 0$.

We first define the following set.
\begin{eqnarray*}
H_{C,D} & := &
\left\{
\begin{array}{l|l}
(\vM,\vN) &
\begin{array}{l}
\vM\cdot\vone + \vN\cdot\vone < 2\ell(C\cdot\vone)(D\cdot\vone) + 3\ell \ \mbox{and}
\\ 
\mbox{there exists a} \ (C,D)\mbox{-biregular graph of size} \ (\vM,\vN)
\end{array}
\end{array}
\right\}
\end{eqnarray*}
Again, such set can be computed greedily
since the number of $(\vM,\vN)$ 
such that $\vM\cdot\vone+\vN\cdot\vone  < 2\ell(C\cdot\vone)(D\cdot\vone) + 3\ell$
is bounded.

Then, the formula $\bireg_{C,D}(\vX,\vY)$ can be defined inductively
as follows.
When $\ell =1$,
\begin{eqnarray*}
\bireg_{C,D}(\vX,\vY) & := & \bireg_{\vc_1,\vd_1}(\vX,\vY)
\end{eqnarray*}
When $\ell \geq 2$,
\begin{eqnarray*}
\bireg_{C,D}(\vX,\vY) & := & 
\bigvee_{(\vM,\vN) \in H_{C,D}} \vX = \vM \ \wedge \ \vY = \vN
\\
& &
\vee\quad
\bigvee_{1\leq j \leq \ell}
\Big(
\:
\vX\cdot \chi(\vc_j)+\vY\cdot \chi(\vd_j) \geq 2(C\cdot\vone)(D\cdot\vone) +3
\\
& &
\qquad\qquad\qquad\quad
\wedge
\quad
\bireg_{C-\vc_j,D-\vd_j}(\vX,\vY)
\\
& &
\qquad\qquad\qquad\quad
\wedge
\quad
\bireg_{\vc_j,\vd_j}(\vX,\vY)
\:
\Big)
\end{eqnarray*}
where $C-\vc_j$, $D - \vd_j$ denote the matrices $C$ and $D$ without row $j$, respectively.
        
We are going to prove that
there exists a $(C,D)$-biregular graph of size $(\vM,\vN)$
if and only if the statement $\bireg_{C,D}(\vM,\vN)$ holds.
The proof is by induction on $\ell$.
The basis $\ell=1$ has been established in Theorem~\ref{t: vd-biregular}.
For the induction step,
we assume that it holds for the case of $\ell-1$
and we are going to prove the case $\ell$.

We first prove the ``only if'' direction.
Suppose $G=(U,V,E_1,\ldots,E_{\ell})$ is $(C,D)$-biregular of size $(\vM,\vN)$.
If $(\vM,\vN)\in H_{C,D}$, then $\bireg_{C,D}(\vM,\vN)$ holds.
So suppose $(\vM,\vN) \notin I_{C,D}$ and
$\vM\cdot\vone + \vN \cdot\vone \geq 2\ell(C\cdot\vone)(D\cdot\vone) + 3\ell$.
Since $C$ and $D$ do not contain zero column,
there exists $j \in \{1,\ldots,\ell\}$
such that $\vM\cdot\chi(\vc_j) + \vN\cdot\chi(\vd_j) \geq 2(C\cdot\vone)(D\cdot\vone)+3$.
Moreover,
if $G=(U,V,E_1,\ldots,E_{\ell})$ is $(C,D)$-biregular of size $(\vM,\vN)$,
then $G$ is also $(C-\vc_j,D-\vd_j)$-biregular and $(\vc_j,\vd_j)$-biregular.
By the induction hypothesis,
both $\bireg_{C-\vc_j,D-\vd_j}(\vM,\vN)$ and $\bireg_{\vc_j,\vd_j}(\vM,\vN)$ hold.

We now prove the ``if'' direction.
Suppose $\bireg_{C,D}(\vM,\vN)$ holds.
If $(\vM,\vN)\in H_{C,D}$, then there exists a $(C,D)$-biregular graph of size $(\vM,\vN)$
and we are done.
So suppose $(\vM,\vN) \notin H_{C,D}$.
Hence there exists $j \in \{1,\ldots,\ell\}$
such that 
\begin{eqnarray*}
& &
\vM\cdot \chi(\vc_j) +
\vN\cdot \chi(\vd_j) \geq 2(C\cdot\vone)(D\cdot\vone) +3
\\
& &
\wedge
\quad
\bireg_{C-\vc_j,D-\vd_j}(\vM,\vN)
\quad
\wedge
\quad
\bireg_{\vc_j,\vd_j}(\vM,\vN)
\end{eqnarray*}
For simplicity, we assume that $j = \ell$.
By the induction hypothesis,
there exists a $(C-\vc_{\ell},D-\vd_{\ell})$-biregular 
graph $G_1=(U_1,V_1,E_1,\ldots,E_{\ell-1})$ of size $(\vM,\vN)$,
and by definition, $E_1,\ldots,E_{\ell-1}$ are pairwise disjoint.
By Theorem~\ref{t: vd-biregular}, there exists a $(\vc_{\ell},\vd_{\ell})$-biregular
graph $G_2 = (U_2,V_2,E_{\ell})$ of size $(\vM,\vN)$.
We can assume that $U_1=U_2=U$ and $V_1=V_2=V$ 
since $G_1$ and $G_2$ are of the same size $(\vM,\vN)$.

We are going to combine $G_1$ and $G_2$ into one graph
to get an $\ell$-type $(C,D)$-biregular 
graph $G=(U,V,E_1,\ldots,E_{\ell})$ of size $(\vM,\vN)$.
If $E_{\ell}\cap (E_1\cup \cdots \cup E_{\ell-1}) =\emptyset$,
then the graph $G = (U,V,E_1,\ldots,E_{\ell})$
is the desired $(C,D)$-biregular $\ell$-type graph of size $(\vM,\vN)$,
and we are done.

Now suppose $E_{\ell}\cap (E_1\cup \cdots \cup E_{\ell-1}) \neq \emptyset$.
We are going to construct another graph $G_2' = (U,V,E_{\ell}')$
such that
\begin{eqnarray*}
|E_{\ell}'\cap (E_1\cup \cdots \cup E_{\ell-1})|
& < &
|E_{\ell}\cap (E_1\cup \cdots \cup E_{\ell-1})|
\end{eqnarray*}
We do this repeatedly until at the end
we obtain a graph $G_2'' = (V,E_{\ell}'')$ such that
$E_{\ell}''\cap (E_1\cup \cdots \cup E_{\ell-1}) = \emptyset$.

Let $(u,v) \in E_{\ell}\cap (E_1\cup \cdots \cup E_{\ell-1})$.
The number of vertices reachable in from $u$ and $v$
within distance $2$ (by any of edges in $E_1,\ldots,E_{\ell}$) 
is $\leq 2(C\cdot\vone)(D\cdot\vone)+2$.
Since $\vM\cdot\chi(\vc_{\ell})+\vN\cdot\chi(\vd_{\ell}) \geq 2(C\cdot\vone)(D\cdot\vone)+3$,
there exists $(u',v') \in E_{\ell} $ such that
$(u,u'), (v,v') \notin E_1\cup\cdots\cup E_{\ell}$.
See the left side of the illustration below.

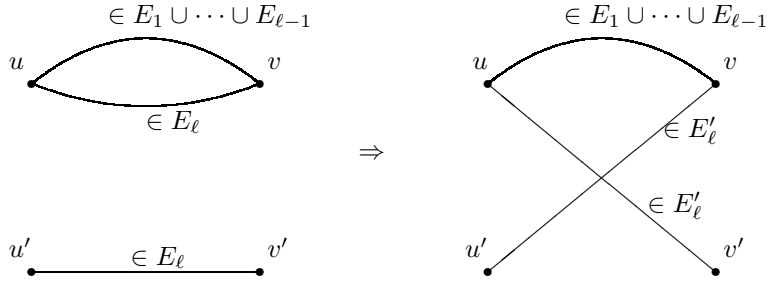
\begin{picture}(200,42)(-103,-19)

\put(-93,12){$u$}
\put(-90,10){\circle*{1}}
\qbezier(-90,10)(-75,22)(-60,10)
\put(-80,18){$\in E_1\cup\cdots\cup E_{\ell-1}$}
\qbezier(-90,10)(-75,4)(-60,10)
\put(-75,4){$\in E_{\ell}$}
\put(-59,12){$v$}
\put(-60,10){\circle*{1}}

\put(-93,-13){$u'$}
\put(-90,-15){\circle*{1}}
\put(-77,-14){$\in E_{\ell}$}
\put(-90,-15){\line(1,0){30}}
\put(-59,-13){$v'$}
\put(-60,-15){\circle*{1}}

\put(-47,0){$\Rightarrow$}

\put(1,12){$v$}
\put(0,10){\circle*{1}}
\put(-20,18){$\in E_1\cup\cdots\cup E_{\ell-1}$}
\qbezier(0,10)(-15,22)(-30,10)
\put(-32,12){$u$}
\put(-30,10){\circle*{1}}
\put(0,10){\line(-6,-5){30}}
\put(-7,3){$\in E_{\ell}'$}

\put(1,-13){$v'$}
\put(0,-15){\circle*{1}}
\put(-9,-7){$\in E_{\ell}'$}
\put(0,-15){\line(-6,5){30}}
\put(-33,-13){$u'$}
\put(-30,-15){\circle*{1}}

\end{picture}

Now we define $E_{\ell}'$ by deleting the edges $(u,v),(u',v')$ from $E_{\ell}$,
while adding the edges $(u,v'),(u',v)$ into $E_{\ell}$.
Formally,
\begin{eqnarray*}
E_{\ell}' & := &
(E_{\ell} - \{(u,v),(u',v')\}) \cup \{(u,u'),(v,v')\}
\end{eqnarray*}
See the right side of the illustration above.

Now it is straightforward that
$G' = (U,V,E_{\ell}')$ is still a $\vd_{\ell}$-regular graph of size $\vN$,
while
\begin{eqnarray*}
|E_{\ell}'\cap (E_1\cup \cdots \cup E_{\ell-1})|
& < &
|E_{\ell}\cap (E_1\cup \cdots \cup E_{\ell-1})|
\end{eqnarray*}
We perform this operation until $E_{\ell+1}\cap (E_1\cup \cdots \cup E_{\ell})=\emptyset$.
This completes the proof of Theorem~\ref{t: l-type vd-biregular}.
\end{proof}

\subsection{For $C \in \bbB^{\ell\times n}$ and $D\in\bbB^{\ell\times m}$
when the number of vertices is ``big enough''}
\label{ss: ugeq biregular}

Let $C \in \bbB^{\ell\times m}$
and $D \in \bbB^{\ell\times n}$,
where $\vc_1,\ldots,\vc_{\ell}$ and $\vd_1,\ldots,\vd_{\ell}$
are the row vectors of $C$ and $D$, respectively.

Two vectors $\vM = (M_1,\ldots,M_m)\in \bbN^m$ 
and $\vN = (N_1,\ldots,N_n) \in \bbN^n$
are ``big enough'' with respect to $C,D$, if
the following inequalities hold.
\begin{eqnarray}
\label{eq: big1}
& & \bigwedge_{1\leq i \leq \ell} 
\quad\Bigg(
\vM\cdot \chi(\vc_i) + \vN\cdot\chi(\vd_i)
\quad\geq\quad 2(C\cdot\vone)(D\cdot\vone) +3
\Bigg)
\\
\label{eq: big2}
 \wedge \qquad & &
\bigwedge_{1\leq i \leq \ell}  \quad
\Bigg(
\Big(
\sum_{j \ \mbox{\scriptsize such that}\ D_{i,j} \in \ugeqN} \ N_j 
\Big)
\quad \geq\quad  \max_{1 \leq j \leq m}(C_{i,j})
\Bigg)
\\
\label{eq: big3}
 \wedge \qquad & &
\bigwedge_{1\leq i \leq \ell}  \quad
\Bigg(
\Big(
\sum_{j \ \mbox{\scriptsize such that}\ C_{i,j} \in \ugeqN} \ M_j 
\Big)
\quad \geq \quad  \max_{1 \leq j \leq m}(D_{i,j})
\Bigg)
\end{eqnarray}

We need a few notations.
In the following for a positive integer $\ell$,
$\id_{\ell}$ denotes the $(\ell\times \ell)$ identity matrix.
For ${\ugeq d} \in {\ugeq\bbN}$, we write $\floor{\ugeq d}$ to denote the number $d$.
By default, we set $\floor{d}=d$.
We also define the $+$ operations on $\bbB$ as follows.
\begin{eqnarray*}
d_1 + d_2 & = & (d_1+d_2)
\\
{\ugeq d_1} + d_2 & = & 
d_1 + {\ugeq d_2} \ = \ 
{\ugeq d_1} + {\ugeq d_2} \ = \
{\ugeq {(d_1+d_2)}}
\end{eqnarray*}
We extend $\floor {\cdot}$ and $+$ to vectors and matrices over $\bbB$
in the natural way, where they are applied componentwise.
For two vectors $\vt_1,\vt_2 \in \bbB^m$,
we define the dot product $\vt_1\cdot\vt_2$ as 
$\floor{\vt_1}\cdot\floor{\vt_2}$.
For a matrix 
$D  \in \bbB^{\ell\times m}$,
we write $D \cdot \vone$ to denote the sum $\sum_{i,j} \floor{D_{i,j}}$.


The lemma below characterises the existence of $(C,D)$-biregular graph
of size $(\vM,\vN)$,
where $\vM,\vN$ are big enough with respect to $(C,D)$
and that for every row $i$,
either the row-$i$ of $C$ or of $D$
contains only elements from $\bbN$.

\begin{lemma}
\label{l: ugeq biregular}
Let 
$C \in 
\left(
\begin{array}{c}
C^{(1)}
\\
C^{(2)}
\\
C^{(3)}
\end{array}
\right) \in \bbB^{\ell\times m}$ 
and 
$D \in 
\left(
\begin{array}{c}
D^{(1)}
\\
D^{(2)}
\\
D^{(3)}
\end{array}
\right) \in \bbB^{\ell\times n}$ 
where $\ell = \ell_1+\ell_2+\ell_3$ and
\begin{itemize}\itemsep=0pt
\item
$C^{(1)} \in \bbN^{\ell_1\times m}$ and $D^{(1)} \in \bbN^{\ell_1\times n}$;
\item
$C^{(2)} \in \bbN^{\ell_2\times m}$ and $D^{(2)} \in \bbB^{\ell_2\times n}$
and every row in $D^{(2)}$ contains an element of ${\ugeq\bbN}$;
\item
$C^{(3)} \in \bbB^{\ell_3\times m}$ and $D^{(3)} \in \bbN^{\ell_3\times n}$
and every row in $C^{(3)}$ contains an element of ${\ugeq\bbN}$.
\end{itemize}
Let $\vM \in \bbN^{m}$ and $\vN \in \bbN^n$ be big enough with respect to $C$ and $D$.
Then the following holds.
There is a $(C,D)$-biregular graph of size $(\vM,\vN)$
if and only if 
$\bireg_{C',D'}((\vM,K_1,\ldots,K_{\ell_3}),(\vN,L_1,\ldots,L_{\ell_2}))$ holds,
where 
\begin{itemize}
\item
$C' = 
\left(\begin{array}{c|c} 
C^{(1)} & 0 
\\ \hline 
C^{(2)} & 0 
\\ \hline
\floor{C^{(3)}} & \id_{\ell_3}
\end{array}\right) \in \bbN^{\ell\times (m+\ell_3)}$
and
$D' =
\left(\begin{array}{c|c} 
D^{(1)} & 0 
\\ \hline 
\floor{D^{(2)}} & \id_{\ell_2}  
\\ \hline
D^{(3)} & 0
\end{array}\right) \in \bbN^{\ell\times (n+\ell_2)}$
\item
each
$K_i = \vd_{\ell_1+\ell_2+i}\cdot \vN - \floor{\vc_{\ell_1+\ell_2+i}}\cdot \vM$,
\item
each
$L_i = \vc_{\ell_1+i}\cdot \vM - \floor{\vd_{\ell_1+i}}\cdot \vN$.
\end{itemize}
\end{lemma}
\begin{proof}
Let 
$C \in 
\left(
\begin{array}{c}
C^{(1)}
\\
C^{(2)}
\\
C^{(3)}
\end{array}
\right) \in \bbB^{\ell\times m}$ 
and 
$D \in 
\left(
\begin{array}{c}
D^{(1)}
\\
D^{(2)}
\\
D^{(3)}
\end{array}
\right) \in \bbB^{\ell\times n}$ and
$\ell_1,\ell_2,\ell_3$,
$\vM \in \bbN^{m}$ and $\vN \in \bbN^n$ be as in the premises.
We also assume that $\vc_1,\ldots,\vc_{\ell}$
and $\vd_1,\ldots,\vd_{\ell}$
are the row vectors of $C$ and $D$, respectively.

Before we present our proof,
we have to remark here that
we do not need the condition that $\vM$ and $\vN$
are big enough to establish the ``only if'' direction.
For the ``if'' direction, we only need Inequalities~\ref{eq: big2} and~\ref{eq: big3}.
Inequality~\ref{eq: big1} is needed only to established Theorem~\ref{t: ugeq biregular}.

We start with the ``only if'' direction.
Suppose $G=(U,V,E_1,\ldots,E_{\ell})$ 
is a $(C,D)$-biregular graph of size $(\vM,\vN)$.
This means there exist a partition $U = U_1 \cup \cdots \cup U_m$
and a partition $V = V_1\cup \cdots \cup V_n$ such that
for each $i =1,\ldots,\ell$,
\begin{itemize}\itemsep=0pt
\item
for each $j =1,\ldots,m$, 
for each $u \in U_j$,
$\deg_{E_i}(u) = C_{i,j}$;
\item
for each $j =1,\ldots,n$, 
for each $v \in V_j$,
$\deg_{E_i}(v) = D_{i,j}$.
\end{itemize}
Now, the following holds.
\begin{itemize}\itemsep=0pt
\item
For each $i = \ell_1+1,\ldots,\ell_1+\ell_2$,
the number of $E_i$-edges in $G$ is $\vc_i\cdot \vM$,
which should be greater than $\floor{\vd_i}\cdot \vN$.
We set $L_{i-\ell_1}= \vc_i\cdot \vM - \floor{\vd_i}\cdot \vN$.
\item
For each $i = \ell_1+\ell_2+1,\ldots,\ell_1+\ell_2+\ell_3$,
the number of $E_i$-edges in $G$ is $\vd_i\cdot \vN$,
which should be greater than $\floor{\vc_i}\cdot \vM$.
We set $K_{i-\ell_1-\ell_2}= \vd_i\cdot \vN - \floor{\vc_i}\cdot \vM$.
\end{itemize}
Let $C',D'$ be as defined in the lemma,
and $\vK = (K_1,\ldots,K_{\ell_3})$
and $\vL = (L_1,\ldots,L_{\ell_2})$.
We construct a $(C',D')$-biregular graph of size 
$((\vM,\vK),(\vN,\vL))$ as follows.
\begin{itemize}\itemsep=0pt
\item
For each $i=\ell_1+1,\ldots,\ell_1+\ell_2$,
for each $j = 1,\ldots,n$,
for each vertex $v \in V_j$,
if $\deg_{E_i}(v) > \floor{D_{i,j}}$,
then we ``split'' $v$ into $v_0,v_1,\ldots,v_k$ vertices,
where 
\begin{itemize}\itemsep=0pt
\item
$k = \deg_{E_i}(v) - \floor{D_{i,j}}$,
\item
$\deg_{E_i}(v_0) = \floor{D_{i,j}}$,
and for each $i' \neq i$,
$\deg_{E_{i'}}(v_0) = \deg_{E_{i'}}(v)$,
\item
$\deg_{E_i}(v_1) = \deg_{E_i}(v_2) = \cdots = \deg_{E_i}(v_k) = 1$,
and for each $i' \neq i$,
$\deg_{E_{i'}}(v_1) = \deg_{E_{i'}}(v_2) = \cdots = \deg_{E_{i'}}(v_k) = 0$.
\end{itemize}

\item
Similarly, for each $i=\ell_1+\ell_2+1,\ldots,\ell_1+\ell_2+\ell_3$,
for each $j = 1,\ldots,m$,
for each vertex $u \in U_j$,
if $\deg_{E_i}(u) > \floor{C_{i,j}}$,
then we ``split'' $u$ into $u_0,u_1,\ldots,u_k$ vertices,
where 
\begin{itemize}\itemsep=0pt
\item
$k = \deg_{E_i}(u) - \floor{D_{i,j}}$,
\item
$\deg_{E_i}(u_0) = \floor{D_{i,j}}$,
and for each $i' \neq i$,
$\deg_{E_{i'}}(u_0) = \deg_{E_{i'}}(u)$,
\item
$\deg_{E_i}(u_1) = \deg_{E_i}(u_2) = \cdots = \deg_{E_i}(u_k) = 1$,
and for each $i' \neq i$,
$\deg_{E_{i'}}(u_1) = \deg_{E_{i'}}(u_2) = \cdots = \deg_{E_{i'}}(u_k) = 0$.
\end{itemize}
\end{itemize}
It should be obvious that the resulting graph is 
a $(C',D')$-biregular graph of size 
$((\vM,\vK),(\vN,\vL))$.

Now we prove the ``if'' direction.
Suppose $\bireg_{C',D'}((\vM,\vK),(\vN,\vL))$ holds,
where each $L_i = \vc_i\cdot \vM - \floor{\vd_i}\cdot \vN$
and $K_i = \vd_i\cdot \vN - \floor{\vc_i}\cdot \vM$.

By Theorem~\ref{t: l-type vd-biregular},
there exists a $(C',D')$-biregular graph $G$ of 
size $((\vM,\vK),(\vN,\vL))$.
Let $G = (U\cup A,V\cup B,E_1,\ldots,E_{\ell})$, where 
$U \cup A= U_1\cup\cdots\cup U_{m}\cup A_1\cup \cdots \cup A_{\ell_3}$ and
$V \cup B= V_1\cdots \cdots \cup V_n \cup B_1\cup\cdots \cup B_{\ell_2}$
are the witness of the $(C',D')$-biregularity.

To construct a $(C,D)$-biregular graph of size $(\vM,\vN)$,
we do the following.
For each vertex $u \in U$ adjacent by $E_i$-edges to, say, $s$ vertices in $B$,
we pick $s$ vertices $v_1,\ldots,v_s$ from the set 
$$
\bigcup_{j \ \mbox{\scriptsize such that}\ d_{i,j} \in \ugeqN} \ V_j 
$$
Such $s$ vertices exist since by Inequality~\ref{eq: big2},
$\sum_{j \ \mbox{\scriptsize such that}\ d_{i,j} \in \ugeqN} \ N_j$
is $\geq \max(\vc_i) \geq \deg(u)$.
We delete those $s$ vertices in $B$,
and connect $u$ to each of $v_1,\ldots,v_s$ by $E_i$-edges.
We do this until the set $B$ is empty.
Similarly, by Inequality~\ref{eq: big3},
we can perform similar operations until the set $A$ is empty.
The resulting graph is a $(C,D)$-biregular graph of size $(\vM,\vN)$.
This completes the proof of Lemma~\ref{l: ugeq biregular}.
\end{proof}

Now Lemma~\ref{l: ugeq biregular} tells us the Presburger formula $\bireg_{C,D}$
for a pair of matrices satisfying the assumption given in Lemma~\ref{l: ugeq biregular}.
More formally, let  
$C \in 
\left(
\begin{array}{c}
C^{(1)}
\\
C^{(2)}
\\
C^{(3)}
\end{array}
\right) \in \bbB^{\ell\times m}$ 
and 
$D \in 
\left(
\begin{array}{c}
D^{(1)}
\\
D^{(2)}
\\
D^{(3)}
\end{array}
\right) \in \bbB^{\ell\times n}$ 
where $\ell = \ell_1+\ell_2+\ell_3$ and
\begin{itemize}\itemsep=0pt
\item
$C^{(1)} \in \bbN^{\ell_1\times m}$ and $D^{(1)} \in \bbN^{\ell_1\times n}$;
\item
$C^{(2)} \in \bbN^{\ell_2\times m}$ and $D^{(2)} \in \bbB^{\ell_2\times n}$
and every row in $D^{(2)}$ contains an element of ${\ugeq\bbN}$;
\item
$C^{(3)} \in \bbB^{\ell_3\times m}$ and $D^{(3)} \in \bbN^{\ell_3\times n}$
and every row in $C^{(3)}$ contains an element of ${\ugeq\bbN}$.
\end{itemize}
That is, for such $C,D$,
we let $\bireg_{C,D}(\vX,\vY)$
as follows.
\begin{eqnarray}
\label{eq: formula}
\bireg_{C,D}(\vX,\vY) & := &
\exists Z_1 \cdots \exists Z_{\ell_3}
\
\exists Z_1' \cdots \exists Z_{\ell_2}'
\\
\nonumber
& & \qquad\quad
\bireg_{C',D'} (\vX,Z_1,\ldots,Z_{\ell_3}, \vY,Z_1',\ldots,Z_{\ell_2}')
\end{eqnarray}
where 
$C' = 
\left(\begin{array}{c|c} 
C^{(1)} & 0 
\\ \hline 
C^{(2)} & 0 
\\ \hline
\floor{C^{(3)}} & I_{\ell_3}
\end{array}\right) \in \bbN^{\ell\times (m+\ell_3)}$
and
$D' =
\left(\begin{array}{c|c} 
D^{(1)} & 0 
\\ \hline 
\floor{D^{(2)}} & I_{\ell_2}  
\\ \hline
D^{(3)} & 0
\end{array}\right) \in \bbN^{\ell\times (n+\ell_2)}$.
Since $C',D'$ consist of entirely $\bbN$ entries,
$\bireg_{C',D'}$ is defined as in Theorem~\ref{t: l-type vd-biregular}.
Intuitively, the variables $Z_1,\ldots,Z_{\ell_3}$
are to capture the values $K_1,\ldots,K_{\ell_3}$
and $Z_1',\ldots,Z_{\ell_2}$ the values $L_1,\ldots,L_{\ell_2}$,
as stated in Lemma~\ref{l: ugeq biregular}.

Using this, we can prove the following theorem.
\begin{theorem}
\label{t: ugeq biregular}
For every
$C \in \bbB^{\ell\times m}$ 
and 
$D \in \bbB^{\ell\times n}$,
there is a Presburger formula $\gbireg_{C,D}(\vX,\vY)$
such that for every $\vM,\vN$ big enough with respect to $C,D$,
the following holds.
There exists a $(C,D)$-biregular graph of size $(\vM,\vN)$
if and only if 
the statement 
$\gbireg_{C,D}(\vM,\vN)$ holds.
\end{theorem}
\begin{proof}
Let $C \in \bbB^{\ell\times m}$ 
and 
$D \in \bbB^{\ell\times n}$,
where $\vc_1,\ldots,\vc_{\ell}$ and $\vd_1,\ldots,\vd_{\ell}$
are the row vectors of $C$ and $D$, respectively.

We need an additional notation.
For a set $I \subseteq \{1,\ldots,\ell\}$,
we write $C(I)$ be the matrix $C'$,
in which each row vector $\vc_i'$ is defined
as $\vc_i' = \floor{\vc_i}$, if $i \in I$,
and $\vc_i' = \vc_i$, if $i \not\in I$.
We can define $D(I)$ similarly.

We define the formula $\bireg_{C,D}(\vX,\vY)$ as follows.
\begin{eqnarray*}
\gbireg_{C,D}(\vX,\vY) & := &
\bigvee_{I_0,I_1,I_2}
\left(
\begin{array}{l}
\bigwedge_{i \in I_0} \vX \cdot \floor{\vc_i} = \vY\cdot \floor{\vd_i} \quad \wedge
\\
\bigwedge_{i \in I_1} \vX \cdot \floor{\vc_i} > \vY\cdot \floor{\vd_i} \quad \wedge
\\
\bigwedge_{i \in I_2} \vX \cdot \floor{\vc_i} < \vY\cdot \floor{\vd_i} \quad \wedge
\\
\bireg_{C(I_0\cup I_2), D(I_0\cup I_1)}(\vX,\vY)
\end{array}
\right)
\end{eqnarray*}
where each $\bireg_{C(I_0\cup I_2), D(I_0\cup I_1)}(\vX,\vY)$
is as defined in Equation~\ref{eq: formula} and
$I_0,I_1,I_2$ range over the partition $I_0\cup I_1\cup I_2 = \{1,2,\ldots,\ell\}$.
Obviously, for every partition $I_0 \cup I_1 \cup I_2 = \{1,\ldots,\ell\}$,
on each row $i=1,\ldots,\ell$,
either row-$i$ from $C(I_0\cup I_2)$,
or row-$i$ from $D(I_0\cup I_1)$,
or row-$i$ from both
consists entirely of $\bbN$.
(Our intention is the application of Lemma~\ref{l: ugeq biregular} later on.)

We are going to prove that the formula $\gbireg_{C,D}(\vX,\vY)$
is the desired formula.
The ``if'' direction follows from Lemma~\ref{l: ugeq biregular}
and that every $C(I_0\cup I_2), D(I_0\cup I_1)$-biregular graph
is obviously also a $(C,D)$-biregular graph.

Now we prove the ``only if'' direction.
Suppose $\vM,\vN$ are big enough for $C,D$.
Let $G=(U,V,E_1,\ldots,E_{\ell})$ be a $(C,D)$-biregular graph of size $(\vM,\vN)$, 
where $U = U_1\cup\cdots\cup U_m$
and $V = V_1 \cup\cdots\cup V_n$
be the witness of the $C,D$-biregularity.

We pick the following paritition $I_0\cup I_1\cup I_2=\{1,\ldots,\ell\}$.
\begin{eqnarray*}
I_0 & = & \{ i \mid \vM \cdot \floor{\vc_i} = \vM\cdot \floor{\vd_i} \}
\\
I_1 & = & \{ i \mid \vM \cdot \floor{\vc_i} > \vM\cdot \floor{\vd_i} \}
\\
I_2 & = & \{ i \mid \vM \cdot \floor{\vc_i} < \vM\cdot \floor{\vd_i} \}
\end{eqnarray*}
We are going to convert the graph $G$ into $(C(I_0\cup I_1),D(I_0,I_2))$-biregular
graph in which $U = U_1\cup\cdots\cup U_m$
and $V = V_1 \cup\cdots\cup V_n$
are also the witness of the $(C(I_0\cup I_1),D(I_0,I_2))$-biregularity.
This, together with Lemma~\ref{l: ugeq biregular}, 
implies that $\bireg_{C(I_0\cup I_1),D(I_0,I_2)}(\vM,\vN)$ holds,
and hence, our theorem.

If $G$ is already a $(C(I_0\cup I_1),D(I_0,I_2))$-biregular graph,
then we are done.
Suppose that $G$ is not.
We do the following three stages.
\\
\underline{\em Stage 1.}
We assume that the following holds for the graph $G$.
For every edge $(u,v)\in E_i$ in $G$, either
\begin{equation}
\label{eq: no-redundant}
\deg_{E_i}(u)= \floor{C_{i,j}}
\quad \mbox{or} \quad
\deg_{E_i}(v)= \floor{D_{i,k}}.
\end{equation}
This can be achieved by doing the following.
Suppose there is an edge $(u,v)\in E_i$
such that $\deg_{E_i}(u) > \floor{C_{i,j}}$
and $\deg_{E_i}(v)> \floor{D_{i,k}}$.
Since $G$ is $(C,D)$-biregular,
this means that $C_{i,j}, D_{i,k} \in {\ugeq\bbN}$.

Deleting the edge $(u,v)$, 
we still have $\deg_{E_i}(u) \geq \floor{C_{i,j}}$
and $\deg_{E_i}(v) \geq \floor{D_{i,k}}$,
and hence $G$ is still $(C,D)$-biregular
with $U = U_1\cup\cdots\cup U_m$
and $V = V_1 \cup\cdots\cup V_n$
be the witness of the $(C,D)$-biregularity.
We repeatedly do this until 
the graph $G$ satisfies condition~(\ref{eq: no-redundant}).
\\
\underline{\em Stage 2.}
We construct a graph $G' = (U \cup S, V \cup T, E_1',\ldots,E_{\ell}')$,
where for every $i=1,\ldots,\ell$,
\begin{itemize}\itemsep=0pt
\item
for every $j=1,\ldots,m$,
for each $u\in U_j$,
$\deg_{E_i'}(u)=\floor{C_{i,j}}$,
\item
for every $s \in S$,
$\deg(s)=1$,\footnote{Recall that $\deg(s)=\deg_{E_1'}(s) + \cdots \deg_{E_{\ell}'}(s)$.
Hence, $\deg(s)=1$ means that there is only one edge adjacent to $s$.}
\item
for every $k=1,\ldots,n$,
for each $v\in V_j$,
$\deg_{E_i'}(v)=\floor{D_{i,j}}$,
\item
for every $t \in T$,
$\deg(t)=1$.
\end{itemize}
The graph $G'$ can be obtained by doing the same trick
as in the proof of Lemma~\ref{l: ugeq biregular}.
For every vertex $u \in U_j$,
if $\deg_{E_i}(u)-\floor{C_{i,j}} = z > 0$,
then we ``split'' $u$ into
$z+1$ vertices $u', s_1,\ldots,s_z$,
where 
\begin{itemize}\itemsep=0pt
\item
$\deg_{E_i}(u') = \floor{C_{i,j}}$,
for all other $h \neq i$,
$\deg_{E_h}(u')=\deg_{E_h}(u)$;
\item
$\deg_{E_i}(s_1)= \cdots = \deg_{E_i}(s_z) = 1$,
and for all other $h \neq i$,
$\deg_{E_h}(s_1)= \cdots = \deg_{E_h}(s_z) = 0$.
\end{itemize}
We can do similar operation to the vertices in $v \in V_k$.
Since $G$ satisfies condition~\ref{eq: no-redundant},
there is no edge between vertices in $S$ and $T$.
We also further partition $S = S_1\cup \cdots \cup S_{\ell}$
and $T = T_1 \cup \cdots \cup T_{\ell}$,
where each $S_i$ and $T_i$ contains the vertices
whose $\deg_{E_i} = 1$.
\\
\underline{\em Stage 3.}
Stage 3 is as follows.
For each $i=1,\ldots,\ell$,
if there are an edge $(s,v) \in E_i$ and an edge $(u,t) \in E_i$,
for some $s\in S_i$, $v\in V_k$, $u\in U_j$, $t \in T_i$,
we do the following.
\begin{itemize}\itemsep=0pt
\item
We delete the two edges $(s,v)$ and $(u,t)$ from $E_i$,
as well as the vertices $s$ and $t$.
\item
We add an edge $(u,v)$ into $E_i$.
\item
If there is already an existing edge $(u,v) \in E_1 \cup \cdots \cup E_{\ell}$,
adding another $(u,v)$ may result in ``parallel'' edges.
However, since $\vM,\vN$ is big enough with respect to $C,D$,
and in particular, Inequality~\ref{eq: big1} holds,
we can apply the same trick as in the proof of Theorem~\ref{t: l-type vd-biregular}
to get rid of the parallel edge,
while preserving the degree of the vertices.
\end{itemize}
We repeatedly do this until
for each $i=1,\ldots,\ell$
either $S_i = \emptyset$, or $T_i=\emptyset$.
In particular, the following holds.
\begin{itemize}\itemsep=0pt
\item
If $i \in I_0$,
then $S_i = T_i =0$.

Recall that $i \in I_0$ means that 
$\vM\cdot \floor{\vc_i} = \vN \cdot \floor{\vd_i}$,
which implies that the initial sets $S_i,T_i$ have the same cardinality.
Since we always delete a pair of vertices $s,t$ from $S_i,T_i$, respectively,
we have at the end $S_i=T_i =\emptyset$.

\item
Likewise, if $i \in I_1$,
then $S_i =0$.

This is because $\vM\cdot \floor{\vc_i} > \vN \cdot \floor{\vd_i}$,
implies that initially $|T_i| > |S_i|$,
which further implies that at the end $S_i=\emptyset$.

By symmetrical reasoning, if $i \in I_2$, then $T_i =0$.
\end{itemize}
From here, we will ``merge'' back
the vertices in $T$ with vertices in $V$.
This is done as follows.
For each vertex $u \in U$ adjacent by $E_i$-edges to, say, $z$ vertices in $T$,
we pick $z$ vertices $v_1,\ldots,v_z$ from the set 
$$
\bigcup_{j \ \mbox{\scriptsize such that}\ d_{i,j} \in \ugeqN} \ V_j 
$$
Such $z$ vertices exist by Inequality~\ref{eq: big2}
(because $\vM,\vN$ are big enough w.r.t. $C,D$).
We delete those $z$ vertices in $T$,
and connect $u$ to each of $v_1,\ldots,v_z$ by $E_i$-edges.
We do this until the set $T$ is empty.
In a similar manner, we can merge back the vertices in $S$ with vertices in $U$,
where the existence of the vertices $v_1,\ldots,v_z$
is guaranteed by Inequality~\ref{eq: big3}.

The resulting graph is $(C(I_0\cup I_1),D(I_0,I_2))$-biregular graph,
which by Lemma~\ref{l: ugeq biregular} the formula implies that
$\bireg_{C(I_0\cup I_2), D(I_0\cup I_1)}(\vM,\vN)$ holds.
This completes our proof of Theorem~\ref{t: ugeq biregular}. 
\end{proof}

\subsection{The notion of partial bipartite graphs}
\label{ss: l-type ugeq biregular}

In this subsection we are going to generalise 
Theorem~\ref{t: ugeq biregular} to the case
when it is possible that one of the inequalities~\ref{eq: big1} and~\ref{eq: big2} does not hold.
The idea is that those numbers (for which the inequalities do not hold)
are hard coded into the Presburger formula.
For this, we introduce the notion of {\em partial graph}.

An $\ell$-type {\em partial bipartite graph} is a tuple 
$\cP = (C,D,S,T,f,g)$, where
\begin{itemize}\itemsep=0pt
\item
$C \in \bbB^{\ell\times m}$ and
$D \in \bbB^{\ell\times n}$,
\item
$S$ is a finite set of vertices (possibly empty),
\item
$T$ is a finite set of vertices (possibly empty),
\item
$f: S \times \{E_1,\ldots,E_{\ell}\} \to \bbB$,
\item
$g: T \times \{E_1,\ldots,E_{\ell}\} \to \bbB$.
\end{itemize}
Obviously, if $S$ or $T$ is empty, 
then $f$ or $g$, respectively, is also an ``empty'' function.
In the following the term {\em partial graph}
always means partial bipartite graph. 

A completion of the partial graph 
$\cP = (C,D,S,T,f,g)$
is a bipartite graph $G=(U\cup S,V\cup T,E_1,\ldots,E_{\ell})$ such that
there is a partition $U_1\cup\cdots \cup U_m$ of $U$
and a partition $V_1\cup\cdots\cup V_n$ of $V$ such that
\begin{itemize}\itemsep=0pt
\item
for every $u\in U_j$, $\deg_{E_i}(u) = C_{i,j}$,
\item
for every $v\in V_j$, $\deg_{E_i}(v) = D_{i,j}$,
\item
for every $s \in S$, $\deg_{E_i}(s) = f(s,E_i)$,
\item
for every $t \in T$, $\deg_{E_i}(t) = g(t,E_i)$.
\end{itemize}
When it is clear from the context,
we also call $U = U_1\cup\cdots \cup U_m$
and $V = V_1\cup\cdots\cup V_n$
the witness of the $(C,D)$-biregularity.
Note that when both $S$ and $T$ are empty,
then the completions of the partial graph $\cP$
are simply $(C,D)$-biregular graphs.

We need a few additional notations.
\begin{eqnarray*}
({\ugeq c}) - d & = &
\left\{
\begin{array}{ll}
{\ugeq {(c-d)}} & \mbox{if} \ c \geq d
\\
{\ugeq 0} & \mbox{otherwise}
\end{array}
\right.
\end{eqnarray*}
Let $C \in \bbB^{\ell\times m}$.
We define a matrix $\xi(C)\in \bbB^{\ell \times (\ell+1)m}$
as follows.
\begin{eqnarray*}
\xi(C) & := &
\Big(
C \mid M_1 \mid \cdots \mid M_{m}
\Big)
\end{eqnarray*}
where each $M_i$ is the matrix obtained by repeating the $i$th column vector of $C$
for $\ell$ number of times, and substracting the identity matrix $\id_{\ell}$.
Formally,
\begin{eqnarray*}
M_i & := &
\left(
\begin{array}{cccc}
C_{1,i} & C_{1,i} & \cdots & C_{1,i}
\\
C_{2,i} & C_{2,i} & \cdots & C_{2,i}
\\
\vdots & \vdots & \ddots & \vdots
\\
C_{\ell,i} & C_{\ell,i} & \cdots & C_{\ell,i}
\end{array}
\right) - \id_{\ell}
\end{eqnarray*}

Lemma~\ref{l: partial-graph-step} below
essentially states that
every partial graph can be reduced into a ``smaller''
partial graph with the addition of some linear equalities.
\begin{lemma}
\label{l: partial-graph-step}
Let $\cP = (C,D,S,T,f,g)$ be a partial graph,
where $T \neq \emptyset$.
Let $t \in T$.
Then the following holds.
\begin{enumerate}[(1)]\itemsep=0pt
\item
For every completion graph $G = (U \cup S, V \cup T,E_1,\ldots,E_{\ell})$
of the partial graph $\cP$
with $U = U_1\cup\cdots \cup U_m$
and $V = V_1\cup\cdots\cup V_n$ being the witness of the $(C,D)$-biregularity,
there exists a completion graph 
$G' = (U \cup S, (V \cup T)\setminus \{t\},E_1',\ldots,E_{\ell}')$
of the partial graph $\cP' = (\xi(C),D,S,T\setminus \{t\},f,g')$,
with the witness of the $(\xi(C),D)$-biregularity being
\begin{eqnarray*}
U & = & U_1' \cup\cdots \cup U_m' 
\quad\cup
\\
& &
(U_{1,1}'\cup \cdots \cup U_{\ell,1}')
\;\cup\;
(U_{1,2}'\cup \cdots \cup U_{\ell,2}')
\;\cup\;
\cdots
\;\cup\;
(U_{1,m}'\cup \cdots \cup U_{\ell,m}')
\\
V & = & V_1 \cup\cdots \cup V_n
\end{eqnarray*}
and
\begin{eqnarray*}
|U_j| & =& |U_j'|+ |U_{1,j}'|+\cdots + |U_{\ell,j}'|
\qquad \mbox{for each} \ j =1,\ldots,m
\\
\sum_{1 \leq j \leq m} |U_{i,j}'| & = & g(t,E_i)
\qquad\qquad\qquad\qquad\quad\;\:
\mbox{for each} \ i =1 ,\ldots, \ell.
\end{eqnarray*}
\item
Visa versa,
for every completion graph 
$G' = (U \cup S, (V \cup T)\setminus \{t\},E_1',\ldots,E_{\ell}')$
of the partial graph $\cP' = (\xi(C),D,S,T\setminus \{t\},f,g')$,
with the witness of the $(\xi(C),D)$-biregularity being 
\begin{eqnarray*}
U & = & U_1' \cup\cdots \cup U_m' 
\quad\cup
\\
& &
(U_{1,1}'\cup \cdots \cup U_{\ell,1}')
\;\cup\;
(U_{1,2}'\cup \cdots \cup U_{\ell,2}')
\;\cup\;
\cdots
\;\cup\;
(U_{1,m}'\cup \cdots \cup U_{\ell,m}')
\\
V & = & V_1 \cup\cdots \cup V_n
\end{eqnarray*}
and for each $i =1 ,\ldots, \ell$,
$\sum_{1 \leq j \leq m} |U_{i,j}'| = g(t,E_i)$,
there exists a completion graph $G = (U \cup S, V \cup T,E_1,\ldots,E_{\ell})$
of the partial graph $\cP$
with $U = U_1\cup\cdots \cup U_m$
and $V = V_1\cup\cdots\cup V_n$ being the witness of the $(C,D)$-biregularity,
and
\begin{eqnarray*}
|U_j| & =& |U_j'|+ |U_{1,j}'|+\cdots + |U_{\ell,j}'|
\qquad \mbox{for each} \ j =1,\ldots,m.
\end{eqnarray*}
\end{enumerate}
\end{lemma}
\begin{proof}
Let $\cP= (C,D,S,T,f,g)$ be a partial graph,
where $T \neq \emptyset$ and $t \in T$.
First, we prove part~(1).
Let $G = (U \cup S, V \cup T,E_1,\ldots,E_{\ell})$
be a completion graph of $\cP$
with $U = U_1\cup\cdots \cup U_m$
and $V = V_1\cup\cdots\cup V_n$ being the witness of the $(C,D)$-biregularity.

For each $j =1,\ldots,m$,
we partition $U_j$ into
\begin{eqnarray*}
U_j & = & 
U_j' \ \cup \ (U_{1,j}'\cup \cdots \cup U_{\ell,j}'),
\end{eqnarray*}
where 
\begin{itemize}\itemsep=0pt
\item
$U_j'$ be the set of vertices in $U_j$
that are not adjacent to the vertex $t$,
\item
for each $i =1,\ldots,\ell$,
$U_{i,j}'$ is the set of vertices in $U_j$
adjacent to $t$ via $E_i$-edges.
\end{itemize}
Now deleting the vertex $t$ and all its adjacent edges,
we obtain the desired completion graph 
$G' = (U \cup S, (V \cup T)\setminus \{t\},E_1',\ldots,E_{\ell}')$
of $\cP'=(\xi(C),D,S,T\setminus \{t\},f,g')$.

Now we prove part~(2).
Let
$G' = (U \cup S, (V \cup T)\setminus \{t\},E_1',\ldots,E_{\ell}')$
be a completion of the partial graph $\cP'=(\xi(C),D,S,T\setminus \{t\},f,g')$,
with the witness of the $(\xi(C),D)$-biregularity being
\begin{eqnarray*}
U & = & U_1' \cup\cdots \cup U_m' 
\quad\cup
\\
& &
(U_{1,1}'\cup \cdots \cup U_{\ell,1}')
\;\cup\;
(U_{1,2}'\cup \cdots \cup U_{\ell,2}')
\;\cup\;
\cdots
\;\cup\;
(U_{1,m}'\cup \cdots \cup U_{\ell,m}')
\\
V & = & V_1 \cup\cdots \cup V_n
\end{eqnarray*}
and for each $i =1 ,\ldots, \ell$,
$\sum_{1 \leq j \leq m} |U_{i,j}'| = g(t,E_i)$.

The desired completion graph
$G = (U \cup S, V \cup T,E_1,\ldots,E_{\ell})$
of the partial graph $\cP$
can be obtained as follows.
We put the vertex $t$ back inside $T$.
Then, for each $i =1,\ldots,\ell$ and for each $j=1,\ldots,m$,
we connect $t$ with every vertex $u \in U_{i,j}'$
with $E_i$-edge.
This way we obtain the completion graph $G$
with $U = U_1\cup\cdots \cup U_m$
and $V = V_1\cup\cdots\cup V_n$ being the witness of the $(C,D)$-biregularity,
and
\begin{eqnarray*}
|U_i| & =& |U_i'|+ |U_{1,i}'|+\cdots + |U_{\ell,i}'|
\qquad \mbox{for each} \ i =1,\ldots,m.
\end{eqnarray*}
This completes our proof of Lemma~\ref{l: partial-graph-step}.
\end{proof}

Following Lemma~\ref{l: partial-graph-step} above,
we show that every partial graph can be translated into
a Presburger formula that captures any of its completion, 
as stated in the following theorem.
\begin{theorem}
\label{t: presburger for partial-graph}
For every partial graph $\cP =(C,D,S,T,f,g)$,
we can construct a Presburger formula $\Psi_{\cP}(\vX,\vY)$
such that 
for every $\vM$ and $\vN$ big enough w.r.t. $C,D$,
the following holds.
There exists a completion graph $G=(U \cup S, V\cup T, E_1,\ldots,E_{\ell})$,
such that $U = U_1\cup \cdots \cup U_m$
and $V = V_1\cup \cdots \cup V_n$
and $\vM = (|U_1|,\ldots, |U_m|)$
and $\vN = (|V_1|,\ldots, |V_n|)$ if and only if
$\Psi_{\cP}(\vM,\vN)$ holds.
\end{theorem}
\begin{proof}
Let $\cP =(C,D,S,T,f,g)$ be a partial graph.
If the matrix $C$ is empty,
there are only finitely many completion of $\cP$.
In this case $\Psi_{\cP}$ simply contains the enumeration the sizes of all possible
completions of $\cP$.
We can define $\Psi_{\cP}$ in a similar manner
when $D$ is empty.

Now suppose both the matrices $C$ and $D$ are not empty.
The construction of $\Psi_{\cP}$ is done inductively as follows.
The base case is $S \cup T = \emptyset$,
in which case $\Psi_{\cP}$ is defined as follows.
\begin{eqnarray*}
\Psi_{\cP}(\vX,\vY) & := &
\gbireg_{C,D}(\vX,\vY),
\end{eqnarray*}
where $\gbireg_{C,D}(\vX,\vY)$ is as defined in Theorem~\ref{t: ugeq biregular}.

Towards the induction step,
let $S \cup T \neq \emptyset$.
Suppose $T \neq \emptyset$ and $t \in T$.
(The case when $S \neq \emptyset$ can handled in a symmetrical manner.)

We define $\Psi_{\cP}(\vX,\vY)$ as follows.
\begin{eqnarray*}
\Psi_{\cP}(\vX,\vY) & := &
\exists Z_1 \cdots \exists Z_{m}
\
\exists Z_{1,1} \cdots \exists Z_{\ell,1}
\ 
\exists Z_{1,2} \cdots \exists Z_{\ell,2}
\ \cdots \
\exists Z_{1,m} \cdots \exists Z_{\ell,m}
\\
& &
\quad\wedge\quad
\bigwedge_{1 \leq i \leq m}
X_i = Z_i + Z_{1,i}+\cdots + Z_{\ell,i}
\\
& &
\quad\wedge\quad
\bigwedge_{1 \leq i \leq \ell}
\sum_{1 \leq j \leq m} Z_{i,j}  =  g(t,E_i)
\\
& &
\quad\wedge\quad
\Psi_{\cP'}((Z_1,\ldots,Z_m,Z_{1,1},\ldots,Z_{\ell,1},\ldots,Z_{1,m},\ldots,Z_{\ell,m}),\vY)
\end{eqnarray*}
where $\cP' = (\xi(C),D,S,T\setminus \{t\},f,g')$
and $g'$ is the function $g$ restricted to $T \setminus \{t\}$.

By Theorem~\ref{t: ugeq biregular} in the previous section,
the correctness of the base case is established.
The induction step follows from Lemma~\ref{l: partial-graph-step}, 
and hence, shows that the formula $\Psi_{\cP}$ is the desired formula.
This completes our proof of Theorem~\ref{t: presburger for partial-graph}.
\end{proof}

\subsection{Constructing the formula $\bireg_{C,D}(\vX,\vY)$ for Theorem~\ref{t: l-type biregular}}
\label{ss: formula l-type biregular}

We need the following notions.
Let $C \in \bbB^{\ell\times m}$ and $D\in \bbB^{\ell\times n}$.
We say that a partial graph $\cP = (C',D',S,T,f,g)$
is compatible with $(C,D)$ with respect to 
a subset $I \subseteq \{1,\ldots,m\}$ and a subset $J \subseteq \{1,\ldots,n\}$, and
the partitions
$S = S_1 \cup \cdots \cup S_{m''}$ and 
$T = T_1 \cup \cdots \cup T_{n''}$,
if the following four conditions hold.
\begin{itemize}\itemsep=0pt
\item
$C' \in \bbB^{\ell\times m'}$ is obtained by deleting the columns $I$ in $C$.
\item
$D' \in \bbB^{\ell\times n'}$ is obtained by deleting the columns $J$ in $D$.  
\item
Let $C'' \in \bbB^{\ell\times m''}$ be the matrix whose columns are the columns $I$ in $C$
where $m'' = m - m'$.
The matrix $C''$ is simply the matrix form of the function $f$.
That is, for each $k=1,\ldots,m''$, 
for every vertex $s\in S_k$,
for every $i=1,\ldots,\ell$,
$f(s,E_i)= C_{i,k}''$.
\item
Let $D'' \in \bbB^{\ell\times n''}$ be the matrix whose columns are the columns $J$ in $D$
where $n'' = n - n'$.
The matrix $D''$ is simply the matrix form of the function $g$.
That is, for each $k=1,\ldots,n''$, 
for every vertex $t\in T_k$,
for every $i=1,\ldots,\ell$,
$g(t,E_i)= D_{i,k}''$.
\end{itemize}

For a subset $I \subseteq \{1,\ldots,m\}$,
and a variable vector $\vX = (X_1,\ldots,X_{\ell})$,
we write $\vX_I$ to denote the variables obtained by deleting $X_i$ whenever $i \in I$.
We can define $\vY_{J}$ similarly when $J \subseteq\{1,\ldots,n\}$
and $\vY = (Y_1,\ldots,Y_n)$.

The formula $\bireg_{C,D}(\vX,\vY)$ as as required in Theorem~\ref{t: l-type biregular}
is as follows.
\begin{eqnarray*}
\bireg_{C,D}(\vX,\vY) & :=  &
\bigvee_{\cP} 
\left(
\begin{array}{l}
 \Psi_{\cP}(\vX_{I},\vY_J) \quad \wedge\quad \varphi
\\
\wedge \ X_{i_1} = |S_1| 
\ \wedge \ X_{i_2} = |S_2|
 \ \wedge \ \cdots \ \wedge \
\ \wedge \ X_{i_{m''}} = |S_{m''}| 
\\  
\wedge \ Y_{j_1} = |T_1| 
\ \wedge \ Y_{j_2} = |t_2|
 \ \wedge \ \cdots \ \wedge \
\ \wedge \ Y_{j_{n''}} = |T_{n''}| 
\end{array}
\right)
\end{eqnarray*}
where 
\begin{itemize}\itemsep=0pt
\item
the disjunction ranges over all partial graph $\cP=(C',D',S,T,f,g)$
compatible with $(C,D)$
w.r.t. $I = \{i_1,\ldots,i_{m''}\}$ and $J = \{j_1,\ldots,j_{n''}\}$,
as well as the partitions $S = S_1 \cup \cdots \cup S_{m''}$
and $T = T_1 \cup \cdots \cup T_{n''}$,
\item
the formula $\psi_{\cP}$ is as defined in Theorem~\ref{t: presburger for partial-graph},
\item
$\varphi$ states that $\vX_{I},\vY_J$ are big enough w.r.t. $(C',D')$,
as defined in the Inequalities~(\ref{eq: big1}),~(\ref{eq: big2}) and~(\ref{eq: big3}).
\end{itemize}
The correctness of the formula $\bireg_{C,D}$ follows
immediately from the correctness of the formula $\Psi_{\cP}$
in Theorem~\ref{t: presburger for partial-graph}.
This completes our proof of Theorem~\ref{t: l-type biregular}.

\subsection{Constructing the formula $\biregc_{C,D}(\vX,\vY)$ for Theorem~\ref{t: main biregular}}
\label{ss: formula l-type complete biregular}

We start with Lemma~\ref{l: lower bound for complete bipartite} which 
essentially states that
if there exists a $(C,D)$-biregular-complete graph of ``big enough'' size,
then for every column $i$ in $C$ and every column $j$ in $D$,
there is a row $l$ such that both $C_{l,i}, D_{l,j} \in {\ugeq\bbN}$.
This means that a $(C,D)$-biregular-complete graph $G=(U,V,E_1,\ldots,E_{\ell})$ 
of ``big enough'' size $(\vM,\vN)$,
then we can connect every pair of vertices $u \in U$ and $v \in V$
with one of the edges without violating the $(C,D)$-biregularity. 

\begin{lemma}
\label{l: lower bound for complete bipartite}
Let $G = (U,V,E_1,\ldots,E_{\ell})$ be an $\ell$-type $(C,D)$-biregular graph
and
$U = U_1\cup \cdots\cup U_m$ and
$V = V_1\cup \cdots\cup V_n$ be the witness of the $(C,D)$-biregularity.
Suppose that for each $i,j$, we have
\begin{eqnarray*}
|U_i|, \ \ |V_j|& \geq & \floor{C}\cdot\vone + \floor{D}\cdot\vone + 1.
\end{eqnarray*}
If $G$ is a complete bipartite graph, then for every $i \in \{1,\ldots,m\}$ and $j \in \{1,\ldots,n\}$,
there exists $l \in \{1,\ldots,\ell\}$ such that
both $C_{l,i}, D_{l,j} \in {\ugeq\bbN}$.
\end{lemma}
\begin{proof}
Let $C\in \bbB^{\ell\times m}$ and $D\in \bbB^{\ell\times n}$,
and $G=(U,V,E_1,\ldots,E_{\ell})$ be a $(C,D)$-biregular-complete graph,
where $U=U_1\cup\cdots\cup U_n$ and 
$V=V_1\cup\cdots\cup V_n$ are the witness of the $(C,D)$-biregularity.
Suppose each $|U_i|$ and $|V_j|$ satisfy the inequality above.

For the sake of contradiction, 
we assume that that there exist $i,j \in \{1,\ldots,m\}$
such that for all $l \in \{1,\ldots,\ell\}$,
either $C_{l,i} \in \bbN$ or $D_{l,j} \in \bbN$.
This means that for each $l \in \{1,\ldots,\ell\}$,
the number of $E_l$-edges between $U_i$ and $V_j$
is  $|U_i|C_{l,i}$, if $C_{l,i} \in \bbN$,
or $|V_j|D_{l,j}$, if $D_{l,j} \in \bbN$.
For each $l=1,\ldots,\ell$, 
\begin{eqnarray*}
K_l & = &
\left\{
\begin{array}{ll}
|U_i|C_{l,i} & \mbox{if} \ C_{l,i} \in \bbN,
\\
|V_j|D_{l,j} & \mbox{if} \ D_{l,j} \in \bbN.
\end{array}
\right.
\end{eqnarray*}
Now the total number of edges between $U_i$ and $V_j$ 
must be $\sum_{1\leq l \leq \ell} K_l$,
which must be equal to $|U_i|\times |V_j|$
since $G$ is a complete bipartite graph.

However, from the inequality
\begin{eqnarray*}
|U_i|, \ \ |V_j|& \geq & \floor{C}\cdot\vone + \floor{D}\cdot\vone + 1,
\end{eqnarray*}
a straightforward calculation shows that 
$K$ is strictly less than $|U_i|\times|V_j|$, a contradiction.
Therefore, for every $i \in \{1,\ldots,m\}$ and $j \in \{1,\ldots,n\}$,
there exists $l \in \{1,\ldots,\ell\}$ such that
both $C_{l,i}, D_{l,j} \in {\ugeq\bbN}$.
This completes the proof of our lemma.
\end{proof}

We say that a pair of matrices $(C,D) \in \bbB^{\ell\times m}\times \bbB^{\ell\times n}$
is an {\em easy pair of matrices}, if
for every $i \in \{1,\ldots,m\}$ and $j \in \{1,\ldots,n\}$
there exists $l \in \{1,\ldots,\ell\}$ such that
both $C_{l,i},D_{l,j} \in {\ugeq\bbN}$.

Lemma~\ref{l: easy matrix vd-biregular} says that
if $(C,D)$ is an easy pair of matrices,
then the formula $\bireg_{C,D}(\vX,\vY)$ as defined in Subsection~\ref{ss: formula l-type biregular}
is sufficient as the required formula $\biregc_{C,D}(\vX,\vY)$ in Theorem~\ref{t: main biregular}.

\begin{lemma}
\label{l: easy matrix vd-biregular}
Let $(C,D)$ be an easy pair of matrices,
where $C \in \bbB^{\ell\times m}$ and $D \in \bbB^{\ell\times n}$.
Then the following holds.
There exists a $(C,D)$-biregular-complete graph of size $(\vM,\vN)$
if and only if
$\bireg_{C,D}(\vM,\vN)$ holds.
\end{lemma}
\begin{proof}
The ``only if'' direction follows directly from Theorem~\ref{t: l-type biregular}. 
Now we prove the ``if'' direction.
Suppose $\bireg_{C,D}(\vM,\vN)$ holds.
By Theorem~\ref{t: l-type biregular},
there exists a $(C,D)$-biregular graph $G=(U,V,E_1,\ldots,E_{\ell})$ 
of size $(\vM,\vN)$.
This graph $G$ is not necessarily complete.
So suppose $U = U_1 \cup \cdots \cup U_m$ and $V = V_1\cup \cdots \cup V_n$
is the witness of the $(C,D)$-biregularity.
If $G$ is not complete,
then we perform the following.
For every $u\in U$ and $v \in V$ such that $(u,v)\notin E_1\cup\cdots \cup E_{\ell}$,
we do the following.
\begin{itemize}\itemsep=0pt
\item
Let $u \in U_i$ and $v \in V_j$.
\item
Pick an index $l \in \{1,\ldots,k\}$
such that $C_{l,i},D_{l,j} \in {\ugeq\bbN}$.
\\
(Such an index $l$ exists since $(C,D)$ is an easy pair.)
\item
Connect $u$ and $v$ with an $E_l$-edge.
\end{itemize}
The resulting graph is now complete and still $(C,D)$-biregular.
This completes our proof of Lemma~\ref{l: easy matrix vd-biregular}.
\end{proof}

If $(C,D)$ is not an easy pair, then
by Lemma~\ref{l: lower bound for complete bipartite},
the values in the entries (in $\vX$ and $\vY$)
corresponding to the columns in $C$ and $D$ that make them not an easy pair
must be bounded.
These values can be encoded as partial graphs as described in the previous section.
This completes our proof of Theorem~\ref{t: main biregular}.


\section{Proof of Theorem~\ref{t: main regular-digraph}}
\label{s: proof2}

The proof is by observing that 
the existence of a $(C,D)$-directed-regular graph
of size $\vN$ 
is equivalent to 
the existence of a $(C,D)$-biregular graph of size $(\vN,\vN)$.
We explain it more precisely below.
\begin{itemize}\itemsep=0pt
\item

Suppose $G=(V,E_1,\ldots,E_{\ell})$ is a $(C,D)$-directed-regular
graph of size $\vN$.

Then, for every vertex $v \in V$, we ``split'' it into
two vertices $u$ and $w$ such that
$u$ is only adjacent to the incoming edges of $v$
and 
$w$ to the outgoing edges of $v$.
See the illustration below.
The left-hand side shows the vertex $v$ before the splitting,
and the right-hand side shows the vertices $u$ and $w$ after the splitting.

\noindent
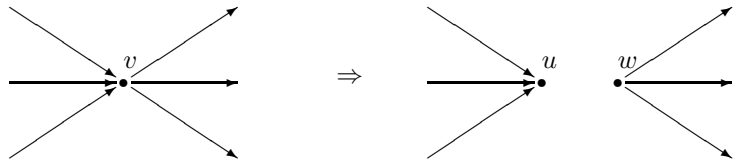
\begin{picture}(200,30)(-110,-15)

\put(-90,10){\vector(3,-2){14}}
\put(-90,0){\vector(1,0){14}}
\put(-90,-10){\vector(3,2){14}}

\put(-75,2){$v$}
\put(-75,0){\circle*{1}}

\put(-74,0.67){\vector(3,2){14}}
\put(-74,0){\vector(1,0){14}}
\put(-74,-0.67){\vector(3,-2){14}}

\put(-47,0){$\Rightarrow$}

\put(-35,10){\vector(3,-2){14}}
\put(-35,0){\vector(1,0){14}}
\put(-35,-10){\vector(3,2){14}}

\put(-20,2){$u$}
\put(-20,0){\circle*{1}}

\put(-10,2){$w$}
\put(-10,0){\circle*{1}}

\put(-9,0.67){\vector(3,2){14}}
\put(-9,0){\vector(1,0){14}}
\put(-9,-0.67){\vector(3,-2){14}}

\end{picture}

Let $U$ be the set of vertices $u$'s
and $W$ the set of vertices $w$'s
after splitting all the vertices in $V$.
Ignoring the orientation of the edges,
the resulting graph is a bipartite graph 
with vertices $U \cup W$
and it is a $(C,D)$-biregular graph of size $(\vN,\vN)$.

\item
Suppose $G=(U,W,E_1,\ldots,E_{\ell})$ is a $(C,D)$-biregular
graph of size $(\vN,\vN)$.
Let $U = U_1\cup \cdots \cup U_m$
and $W = W_1 \cup\cdots \cup W_m$
be the witness of the $(C,D)$-biregularity.
We denote by $U_i = \{u_{i,1},\ldots,u_{i,K_i}\}$
and $W_i = \{w_{i,1},\ldots,w_{i,K_i}\}$
for each $i=1,\ldots,m$.

Now we put the orientation on all the edges from $U$ to $W$.
Then we merge every two vertices $u_{i,j}$ and $w_{i,j}$
into one vertex $v_{i,j}$.
This way, we obtain a $(C,D)$-directed-regular graph $G=(V,E_1,\ldots,E_{\ell})$
with $V = V_1 \cup \cdots \cup V_m$ be the witness of the $(C,D)$-regularity
where $V_{i}= \{v_{i,1},\ldots,v_{i,K_i}\}$ for each $i=1,\ldots,m$.

However, with such merging it is possible that
there is a self-loop $(v,v)$ in $G$
or a pair of edges $(v,v'),(v',v) \in E_1 \cup \cdots \cup E_{\ell}$.
We can get rid of the self-loop $(v,v)$ 
without violating the $(C,D)$-directed-regularity as follows.
The trick is similar to the one used before.
Assuming that the size of each $|V_1|,\ldots,|V_{m}|$ is big enough
and there are enough edges in each $E_1,\ldots,E_{\ell}$,
there is an edge $(v',v'')$ of the same type
such that both $v',v''$ are not adjacent to $v$.
Deleting the edge $(v,v)$ and $(v',v'')$,
and adding the edges $(v',v)$ and $(v,v'')$,
we obtain a $(C,D)$-directed-regular graph with one less self-loop.
We do this repeatedly until there is no more self-loop. 
See the illustration below.

\noindent
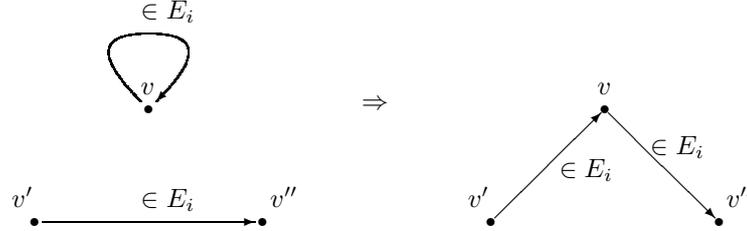
\begin{picture}(200,42)(-115,-27)

\put(-76,-3){$v$}
\put(-75,-5){\circle*{1}}
\put(-76,7){$\in E_i$}
\qbezier(-76,-4)(-85,5)(-75,5)
\qbezier(-74,-4)(-65,5)(-75,5)
\put(-74,-4){\vector(-1,-1){0}}

\put(-93,-18){$v'$}
\put(-90,-20){\circle*{1}}
\put(-76,-18){$\in E_i$}
\put(-89,-20){\vector(1,0){28}}
\put(-59,-18){$v''$}
\put(-60,-20){\circle*{1}}

\put(-47,-5){$\Rightarrow$}

\put(-16,-3){$v$}
\put(-15,-5){\circle*{1}}

\put(-33,-18){$v'$}
\put(-30,-20){\circle*{1}}
\put(-29.5,-19.5){\vector(1,1){14}}
\put(-21,-14){$\in E_i$}

\put(1,-18){$v''$}
\put(0,-20){\circle*{1}}
\put(-14.5,-5.5){\vector(1,-1){14}}
\put(-9,-11){$\in E_i$}

\end{picture}

Similarly, we can get rid of 
a pair of edges $(v,v'),(v',v) \in E_1 \cup \cdots \cup E_{\ell}$
without violating the $(C,D)$-directed-regularity
as follows.
Again, the trick is similar to the one used before.
Assuming that the size of each $|V_1|,\ldots,|V_{m}|$ is big enough
and there are enough edges in each $E_1,\ldots,E_{\ell}$,
there is an edge $(w,w')$ of the same type
such that both $w,w'$ are not adjacent to either $v$ or $v'$.
Deleting the edge $(v,v')$ and $(w,w')$,
and adding the edges $(v,w')$ and $(w,v')$,
we obtain a $(C,D)$-directed-regular graph with one less parallel edges.
We do this repeatedly until there are no more parallel edges. 
See the illustration below.

\noindent
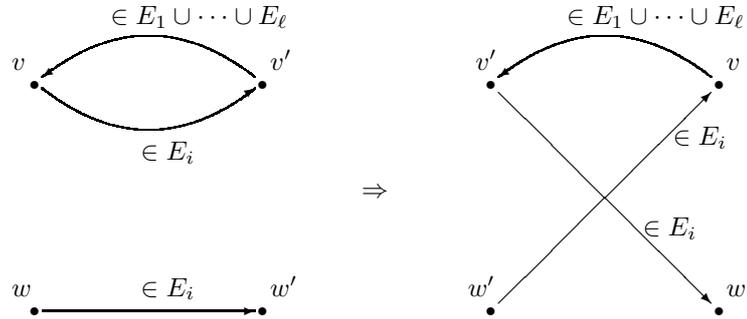
\begin{picture}(200,50)(-115,-25)

\put(-93,12){$v$}
\put(-90,10){\circle*{1}}
\put(-80,18){$\in E_1\cup\cdots\cup E_{\ell}$}
\qbezier(-89,11)(-75,22)(-61,11)
\put(-89,11){\vector(-3,-2){0}}
\qbezier(-89,9.5)(-75,-1.5)(-61,9.5)
\put(-61,9.5){\vector(3,2){0}}
\put(-76,0){$\in E_i$}
\put(-59,12){$v'$}
\put(-60,10){\circle*{1}}

\put(-93,-18){$w$}
\put(-90,-20){\circle*{1}}
\put(-89,-20){\vector(1,0){28}}
\put(-76,-18){$\in E_i$}
\put(-59,-18){$w'$}
\put(-60,-20){\circle*{1}}

\put(-47,-5){$\Rightarrow$}

\put(-32,12){$v'$}
\put(-30,10){\circle*{1}}
\qbezier(-29,11)(-15,22)(-1,11)
\put(-29,11){\vector(-3,-2){0}}
\put(-29,9){\vector(1,-1){28}}
\put(-6,2){$\in E_i$}

\put(1,12){$v$}
\put(0,10){\circle*{1}}
\put(-20,18){$\in E_1\cup\cdots\cup E_{\ell}$}

\put(1,-18){$w$}
\put(0,-20){\circle*{1}}

\put(-10,-10){$\in E_i$}

\put(-33,-18){$w'$}
\put(-30,-20){\circle*{1}}
\put(-29,-19){\vector(1,1){28}}

\end{picture}

If some sets $V_i$ or $|E_i|$ are of a fixed size,
those can be encoded in a partial graph
in the same manner discussed in Subsection~\ref{ss: l-type ugeq biregular}.

\end{itemize}
We omit the technical details since we essentially 
run through the same arguments used in the previous section.

\section{Concluding remarks}
\label{s: conclusion}

In this paper we have shown that the spectra of $\Ctwo$ formulae
are semilinear sets.
The proof is by constructing the Presburger formulae
that express precisely the spectra.
As far as our knowledge is concerned,
the logic $\Ctwo$ is the first logic whose spectra is closed under complement
without any restriction on the vocabulary nor in the interpretation.

Furthermore, from our proof we can easily deduce a few easy corollaries.
The first is that the family of models of a $\QMLC$ formula
can be viewed as a collection of biregular graphs and regular digaphs
in the following sense.
Let $\phi$ be a $\QMLC$ formula 
and $\cR =\{R_1,\ldots,R_{\ell},\rev{R}_1,\ldots,\rev{R}_{\ell}\}$ 
be the set of binary relations used in $\phi$ and
that $\rev{R}_i$ is the inversed relation of $R_i$.
Let $K$ be the integer such that for all 
subformulae $\Diamond_R^l \psi$ in $\phi$,
we have $l \leq K$.

Let $\cT_{\phi}$ is the set of all types in $\phi$.
For a model $\stra \models \phi$,
we partition $A = \bigcup_{T \in \cT_{\phi}} A_T$,
where $A_T = \{a \in A \mid a \ \mbox{is of type} \ T\}$.
That the model $\stra$ is a collection of regular graphs
is in the following sense.
Recall the matrices $D_T$, $\rev{D}_T$,
$D_{S\to T}$ and $D_{T \to S}$
as defined in Section~\ref{s: proof}.

For a type $T$,
by restricting the relations $R_1,\ldots,R_{\ell}$
on the elements in $A_T$,
we obtain a $(D_T,\rev{D}_T)$-regular digraph 
$G_T = (A_T, E_1,\ldots,E_{\ell})$, where
for each $E_i$,
\begin{eqnarray*}
E_i & = & R_{i}\cap (A_T\times A_T)
\end{eqnarray*}
For two different types $S,T$,
by restricting the relations 
on $A_S \times A_T$,
we obtain a $(D_{S\to T},\rev{D}_{S\to T})$-biregular graph 
$G_{S,T} = (A_S,A_T, E_1,\ldots,E_{\ell},E_1',\ldots,E_{\ell}')$, where
each $E_i,E_i'$ are
\begin{eqnarray*}
E_i  =  R_{i}\cap (A_S\times A_T)
& \qquad\mbox{and}\qquad &
E_i'  =  \rev{R}_{i}\cap (A_S\times A_T)
\end{eqnarray*}

Theorem~\ref{t: presburger for c2} can be further generalised as follows.
Let $\cP = (P_1,\ldots,P_l)$, where $P_1,\ P_2,\ \ldots,\ P_l$ are unary
predicates.
Define the {\em image} of a structure $\stra$
as $\parikh_{\cP}(\stra)=(|P_1^{\stra}|,\ldots,|P_l^{\stra}|)$.
We also define the image of a formula $\varphi$ with predicates from $\cP$
as $\parikh_{\cP}(\varphi) = \{\parikh_{\cP}(\stra) | \stra \models \varphi\}$. 
It must be noted here that $P_1^{\stra},\ldots,P_{l}^{\stra}$
are not necessarily disjoint, and that 
they may not cover the whole domain $A$.
For this reason, 
the notion of image is more general than the notion of many-sorted spectrum
which requires the unary predicates to partition the whole domain.
With a slight adjustment in our proof in Section~\ref{s: proof},
we can obtain the following two corollaries.

\begin{corollary}
\label{c: parikh semilinear}
Let $\phi \in \Ctwo$ and
$\cP = (P_1,\ldots,P_l)$, where $P_1,\ldots,P_{l}$ 
be a set of unary predicates in $\phi$.
The set $\{\parikh_{\cP}(\stra) \mid \stra \models \phi\}$ is semilinear.
\end{corollary}

\begin{corollary}
\label{c: parikh decidable}
Let $\cP = (P_1,\ldots,P_l)$.
The following problem is decidable.
Given a $\Ctwo$ formula $\phi$ and a Presburger formula $\Psi(x_1,\ldots,x_l)$,
determine whether there exists a structure $\stra\models \phi$
such that $\Psi(\parikh_{\cP}(\stra))$ holds.
\end{corollary}

There are still a few more questions that we would like to investigate for future work.
The first natural question is: how can be $\Ctwo$ extended while keeping decidability?
Using three variables ($\FOthree$) one can easily encode a grid; therefore,
the satisfiability problem is no longer decidable (and thus the image membership problem). 
However, we could extend $\Ctwo$ by giving access to a relation 
having a property which is undefinable in $\Ctwo$, such as transitivity. 
In particular, $\CtwoT$, that is, the logic $\Ctwo$ with access
to a total order on the universe, seems powerful: Petri net
reachability~\cite{Mayr81, Kosaraju82, Leroux09}
reduces to image membership for $\CtwoT$ formulae. 
We do not know whether a reduction exists in the other direction. 
Another possible extension is to add an equivalence relation to $\Ctwo$.

\section*{Acknowledgment}
We thank the anonymous referees for their comments and suggestions
which greatly improve our paper.
The second author also acknowledges the generous financial support of FWO,
under the scheme FWO Pegasus Marie Curie fellowship.

\bibliographystyle{plain}

\end{document}